\documentclass[12pt,draftcls,journal,onecolumn]{IEEEtran}
\usepackage{amsfonts,color,morefloats,pslatex,graphicx,graphics}
\usepackage{amssymb,amsthm, amsmath,latexsym}

\newcommand{\red}{\color{red}}
\newtheorem{theorem}{Theorem}
\newtheorem{lemma}[theorem]{Lemma}
\newtheorem{remark}[theorem]{Remark}
\newtheorem{proposition}[theorem]{Proposition}

\newtheorem{definition}[theorem]{Definition}
\newtheorem{example}[theorem]{Example}

\newcommand{\bF}{{\mathbb{F}}}

\newcommand{\C}{{\mathcal{C}}}

\usepackage{blindtext}

\ifCLASSINFOpdf

\else

\fi

\begin{document}
%
\title{ Two classes of reducible cyclic codes with large minimum symbol-pair distances
\thanks{ }
}

\author{Xiaoqiang Wang,\,\,\, Yue Su, \,\,\, Dabin Zheng,\,\,\, Wei Lu
\thanks{X. Wang, Y. Su, D. Zheng and W. Lu are with Hubei Key Laboratory of Applied Mathematics, Faculty of Mathematics and Statistics, Hubei University, Wuhan 430062, China, Email: waxiqq@163.com, yuesut@163.com, dzheng@hubu.edu.cn, d20230130@hubu.edu.cn.  The corresponding author is Dabin Zheng.}}

\maketitle

\begin{abstract}
The high-density data storage technology aims to design high-capacity storage at a relatively low cost.
In order to achieve this goal, symbol-pair codes were proposed by Cassuto and Blaum \cite{CB10,CB11} to handle channels that
output pairs of overlapping symbols. Such a channel is called symbol-pair
read channel, which introduce new concept called symbol-pair weight and minimum symbol-pair distance.
In this paper, we consider the parameters of two classes of reducible cyclic codes under the symbol-pair metric. Based on the theory of cyclotomic numbers and Gaussian period over finite fields, we show the possible symbol-pair weights of these codes. Their minimum symbol-pair distances are twice the minimum Hamming distances under some conditions. Moreover, we obtain some three symbol-pair weight codes and determine their symbol-pair weight distribution.
A class of MDS symbol-pair codes is also established. Among other results, we determine the values of some generalized cyclotomic numbers.
\end{abstract}

\IEEEpeerreviewmaketitle

\section{Introduction}\label{sec-auxiliary}

In 1948, in order to to solve the problem about reliable communication in an engineering problem, Shannon \cite{CE48} showed that good codes exist, gave birth to information theory and coding theory.  In 1950,  Hamming weight and Hamming distance were introduce  by Hamming, which gives the basis for classical error-correcting codes. Since than, a lot of works on Hamming metric of linear codes have been given. These works are based on that noisy channels are analyzed by dividing the messages into individual information units.

By the application of high-density data storage technologies, due to physical limitations, an individual symbol cannot be read off the channel and each channel read contains contributions from two adjacent symbols.  In order to
protect against pair errors in symbol-pair read channels,  symbol-pair codes were first studied by Cassuto and Blaum
\cite{CB10,CB11}. The minimum symbol-pair weight and minimum symbol-pair distance of a code over finite field $\bF_q$ were defined as follows:
Let $\mathbf{x} = (x_0, x_1,\cdots, x_{n-1})$ be a vector in $ \bF_q^n$ and
\begin{equation*}
\pi_p(\mathbf{x}) =[(x_0, x_1),(x_1, x_2),\ldots,(x_{n-2}, x_{n-1}),(x_{n-1}, x_0)] \in (\bF_q^2)^n.
\end{equation*}
Two pairs $(c,d)$ and $(e,f)$ are distinct if $c\neq e$ or $d\neq f$, or both. For any two codewords $\mathbf{c_1}$ and $\mathbf{c_2}$, the symbol-pair distance
between $\mathbf{c_1}$ and $\mathbf{c_2}$ is defined as ${d_p}(\mathbf{c_1}, \mathbf{c_2}) = d_H(\pi_p(\mathbf{c_1}),\pi_p(\mathbf{c_2})),$
where ${d_H}(\cdot)$ denotes the usual Hamming distance.
The minimum symbol-pair distance of a code $\C$ is defined to be
$$d_p(\C) = min\{d_p(\mathbf{c}_1,\mathbf{c_2}) \mid \mathbf{c}_1,\mathbf{c_2}\in \C,\,\,\mathbf{c}_1\neq \mathbf{c}_2 \}.$$ Accordingly,  the minimum symbol-pair weight of a code $\C$  is defined as
$$w_p(\C)=\min\{w_p({\bf c})\,|\, {\bf 0}\ne {\bf c}\in \C \},
$$ where
$$w_p({\bf c})=\Big|\{i\,|\,(c_i,c_{i+1})\,{\not=}\,(0,0),\,\, 0 \le i \le n-1, \,\,c_n=c_0\}\Big|.$$

It is clear that any cyclic code $\C$ of length $n$ over $\mathbb{F}_q$ corresponds to a subset of the quotient ring
$\mathbb{F}_q[x]/\langle x^n-1\rangle$. Since every ideal of $\mathbb{F}_q[x]/\langle x^n-1\rangle$ must be principal, $\C$ can be expressed as $\C=\langle g(x)\rangle$, where
$g(x)$ is a monic polynomial with the smallest degree and is called the generator polynomial. Let $h(x)=(x^n-1)/g(x)$, then $h(x)$ is referred to as the check polynomial of $\C$.
The zeros of $g(x)$ and $h(x)$ are called zeros and non-zeros of $\C$. If $g(x)$ is  irreducible over $\mathbb{F}_q$, then $\mathcal{C}$ is called an irreducible cyclic code. Otherwise, $\mathcal{C}$ is called a reducible cyclic code.
Let $A^p_i$ denote the number of nonzero codewords with symbol-pair weight $i$ in $\mathcal{C}$. Similar as Hamming weight enumerator and Hamming weight distribution, the symbol-pair weight enumerator and symbol-pair weight distribution are defined by $1+A^p_1x+A^p_2x^2+\cdots+A^p_nx^n$ and  $(1, A^p_1, \ldots, A^p_n)$, respectively. If the number of nonzero $A^p_i$ in the sequence $(1, A^p_1, \ldots, A^p_n)$ is equal to $t$, then  the code $\mathcal{C}$ is said to be a $t$ symbol-pair weight code.

In recent years, many researchers focused on symbol-pair codes and made a lot of progress on this topic. In \cite{CB10,CB11}, Cassuto and Blaum  provided constructions and decoding methods of
symbol-pair codes. Soon later,  Chee $et$ $al$. \cite{CJKWY13} established the Singleton type bound of symbol-pair codes as follows: Let $\C$ be a symbol-pair code of length $n$ with size $M$ and
minimum pair-distance $d_p$ over $\mathbb{F}_q$, then $M\leq q^{n-d_p+2}$, where $2 \leq d_p \leq n$. If the equality holds, then $\C$ is called an optimal code with respect
to Singleton type bound, or an MDS symbol-pair code. After establishing the Singleton type Bound,  a lot of  MDS symbol-pair codes
are constructed by different kinds of techniques such as algebraic, geometry, graph and so on. The reader is referred to, for example, \cite{Chen17,DingGe18,Dinh22,KZL15,LiGe17,MaLuo22} for information.
The research on other bounds of symbol-pair codes have also received the widespread attention. The Gilbert-Varshamov type bound is established in \cite{CB11} and the
Plotkin-type bound is derived in \cite{Wang20}. Elishco et al. \cite{Elishco20} gave the Johnson-type bound for binary symbol-pair codes with even minimum
symbol-pair distance, then Chen and Liu \cite{Chen23} extended the result to $q$-ary symbol-pair codes with arbitrary minimum symbol-pair distance. In addition, Elishco et al. \cite{Elishco20} gave the linear programming type bound and Yan et al. \cite{Yan20} showed the Elias type upper bound. With these bounds, some new optimal symbol-pair codes were also constructed. The reader is referred to, for example, \cite{Chen23,Elishco20,Wang20} for information. Meanwhile, some decoding algorithms of symbol-pair codes have been studied by many researchers. These work have been developed in several aspects. The reader can refer to \cite{Horii16,Hirotomo14,Liu19,Hirotomo16,Takita17,Yaakobi12,Yaakobi16} and references therein.

In coding theory,  many researchers focused on construction of cyclic codes with few weights and determined their Hamming weight distribution.
However, there seems very few research about determining the symbol-pair weight distribution of cyclic codes since it is a very difficult problem.   To the best of our knowledge, up to now, the contributions on the symbol-pair weight distribution of linear codes are the following:

\begin{itemize}

\item Shi et al. \cite{Shi2021,Zhu2021,Zhu2022} studied the symbol-pair weight distribution of several classes of cyclic codes.

\item Youcef \cite{aouche2023} determined the symbol-pair weight distribution of some of irreducible cyclic codes.

\item Vega \cite{Vega23} showed the symbol-pair weight distribution of all one symbol-pair weight and two  symbol-pair weight semiprimitive irreducible cyclic codes.

\item Ma and Luo \cite{Ma21} gave the symbol-pair weight distribution for MDS codes and simplex codes.

\end{itemize}

In \cite{CB10,CB11}, the authors pointed out that for any code $\mathcal{C}$ of length $n$, the minimum symbol-pair distance $d_p(\C)$ and the minimum Hamming distance $d_H(\C)$ satisfy the relationship $d_H(\C)+ 1 \leq d_p(\C) \leq 2d_H(\C)$ if $d_H(\C)<n$.
So far,  there seems no researchers to consider that for some good cyclic codes under the Hamming metric, which code can
achieve the largest minimum symbol-pair distance compared with their minimum Hamming distance, i.e., $d_p(\C)=2d_H(\C)$. In \cite{CMa2010}, the authors  considered two classes of reducible cyclic codes, and showed that these codes have good parameters,  and closely related to projective sets in finite geometry and strongly regular graphs \cite{Calderbank84} and have applications in
cryptography \cite{Yuan06}. In this paper, we continue studying the cyclic codes proposed in~\cite{CMa2010} under the symbol-pair metric. The possible symbol-pair weights of these codes are obtained and the codes achieves the highest minimum symbol-pair distances compared with their minimum Hamming distances in many cases. In addition, we obtain some three symbol-pair weight codes and determine their symbol-pair weight distribution. By the puncturing technique, a class of MDS symbol-pair codes is also established. Moreover, the values of some generalized cyclotomic numbers are determined.

The rest of this paper is organized as follows. Section II introduces some preliminaries. In Section III, we determine the symbol-pair weight distribution of the first class of cyclic codes. In Section IV, we show the possible symbol-pair weights of the second class of cyclic codes. For some special cases, we obtain their symbol-pair weight distribution.  In Section V, we obtain a class of symbol-pair MDS codes by the puncture technique. Section \ref{sec-finals} concludes the paper.

\section{Preliminaries}\label{sec-auxiliary}

In this section, we introduce some
basic concepts, which will be used later in this paper.
Starting from now on, we adopt the following notation unless otherwise stated:
\begin{itemize}

\item $\mathbb{F}_q$ is the finite field with $q$ elements, where ${q=p^s}$, $p$ is a prime and $s\geq 1$ is a positive integer.

\item $\mathbb{F}_r$ is the finite field with $r$ elements, where ${r=q^m}$ and $m\geq 2$ is a positive integer.

\item $e>1$ and $h$ are positive integers satisfying $e\,\, |\,\,h$ and $h\,\, |\,\,q-1$.

\item $g=\alpha^{(q-1)/ h}$ and $\beta=\alpha^{(r-1)/e}$, where $\alpha$ is a primitive element of $\bF_{r}$.

\item $\chi_1$ and $\chi$ are the canonical additive character of $\bF_{q}$ and $\bF_{r}$, respectively.

\item ${\rm Tr}_u^v(\cdot)$ is the trace function from $\bF_v$ to $\bF_{u}$, where $\bF_v$ and $\bF_{u}$ are finite fields with $v$ and $u$ elements, respectively.

\item $\C_{(q,m,h,e)}$ is the cyclic code of length $n=\frac{h(r-1)}{q-1}$ with check polynomial $m_{g^{-1}}(x)m_{(\beta g)^{-1}}(x)$, where $m_{g^{-1}}(x)$ and $m_{(\beta g)^{-1}}(x)$ are the minimum polynomials of $g^{-1}$ and $(\beta g)^{-1}$, respectively.

\end{itemize}

\subsection{The generic construction of cyclic codes}

In the subsection, a generic construction of cyclic codes with two nonzeros is introduced, which was first given
by Ma et al. in \cite{CMa2010}. Recall that $\C_{(q,m,h,e)}$ is the cyclic code of length $n=\frac{h(r-1)}{q-1}$ with the check polynomial $m_{g^{-1}}(x)m_{(\beta g)^{-1}}(x)$. By the Delsarte's theorem,
the code $\C_{(q,m,h,e)}$ can be expressed as
\begin{equation}\label{cqmhe}
 \C_{(q,m,h,e)}=\left\{{\bf c}(a,b)\,:\,a,b \in \bF_{r}\right\},
\end{equation}
where
$${\bf c}(a,b)=({\rm Tr}_q^r(ag^0+b(\beta g)^0),{\rm Tr}_q^r(ag+b(\beta g)),\ldots,{\rm Tr}_q^r(ag^{n-1}+b(\beta g)^{n-1})).$$
It is clear that the symbol-pair weight of the codeword ${\bf c}(a,b)$  is equal to
\begin{equation}\label{eq:cab}
{\rm wt_p}( {\bf c}(a,b))=n-T(a,b),
\end{equation}
 where
\begin{equation*}
T(a,b)=\left|\left\{0\leqslant i \leqslant n-1\,:\,{\rm Tr}^r_q(ag^i+b(\beta g)^i)=0\,\, \text{and}\,\,{\rm Tr}^r_q(ag^{i+1}+b(\beta g)^{i+1})=0\right\}\right|.
\end{equation*}
Since $\langle g \rangle=\langle\alpha^{(q-1)/h}\rangle$ and $\{g^i\,:\,0\leqslant i \leqslant n-1\}= C_0^{((q-1)/h,r)}$, we have
\begin{equation*}
\begin{aligned}
T(a,b)&=\left|\left\{x \in C_0^{((q-1)/h,r)}\,:\,{\rm Tr}^r_q(ax+b\beta^{log_gx}x)=0\,\, \text{and}\,\,{\rm Tr}^r_q(agx+b\beta^{log_ggx
}gx))=0\right\}\right|\\
&=\tfrac{1}{q^2}\sum_{x \in C_0^{((q-1)/h,r)}}\sum_{z \in \bF_{q}} \chi_1(z{\rm Tr}^r_q(ax+b\beta^{log_gx}x))\sum_{y \in \bF_{q}} \chi_1(y{\rm Tr}^r_q(agx+b\beta^{log_ggx}gx))\\
&=\tfrac{1}{q^2}\sum_{x \in C_0^{((q-1)/h,r)}}\sum_{z \in \bF_{q}}\chi(z(ax+b\beta^{log_gx}x))\sum_{y \in \bF_{q}} \chi(y(agx+b\beta^{log_ggx}gx)).
\end{aligned}
\end{equation*}
It is easy to check that
$$C_0^{((q-1)/h,r)}=\bigcup_{i=0}^{e-1}C_{(q-1)i/h}^{((q-1)e/h,r)},$$
then we obtain
\begin{equation*}
\begin{aligned}
T(a,b)
=&\tfrac{1}{q^2}\sum_{i=0}^{e-1}\sum_{x \in C_{(q-1)i/h}^{((q-1)e/h,r)}}\sum_{z \in \bF_{q}}\chi(axz+\beta^i bxz)\sum_{y \in \bF_{q}}\chi(agxy+\beta^{i+1} bgxy)\\
=&\tfrac{h(r-1)}{q^2(q-1)}+\tfrac{2}{q^2}\sum_{i=0}^{e-1}\sum_{x \in C_{(q-1)i/h}^{((q-1)e/h,r)}}\sum_{z \in \bF_{q}^\ast}\chi(xz(a+\beta^i b))\\
&+\tfrac{1}{q^2}\sum_{i=0}^{e-1}\sum_{x \in C_{(q-1)i/h}^{((q-1)e/h,r)}}\sum_{z,y \in \bF_{q}^\ast}\chi(xz(a(1+gy)+\beta^{i}b(1+\beta gy))).
\end{aligned}
\end{equation*}
Let $d={\rm gcd}(m,\tfrac{e(q-1)}{h})$, from Lemma \ref{lem:GHW} we can get the following multiset equality
\begin{equation*}
\begin{aligned}
\left\{xz:x \in C_{(q-1)i/h }^{((q-1)e/h,r)},z \in \bF_{q}^\ast\right\}=\tfrac{hd}{e}\ast C_{(q-1)i/h}^{(d,r)}.
\end{aligned}
\end{equation*}
Then we have
\begin{equation}\label{eq:Tab}
\begin{split}
T(a,b)=&\tfrac{h(r-1)}{q^2(q-1)}+\tfrac{2hd}{eq^2}\sum_{i=0}^{e-1}\sum_{x \in C_{(q-1)i/h}^{(d,r)}}\chi(x(a+\beta^i b))\\
&+\tfrac{hd}{eq^2}\sum_{i=0}^{e-1}\sum_{x \in C_{(q-1)i/h}^{(d,r)}}\sum_{y \in \bF_{q}^\ast}\chi(x(a(1+gy)+\beta^i b(1+\beta gy))).
\end{split}
\end{equation}
In general, it is very difficult to determine the value distribution of $T(a,b)$ for $(a,b)$ running through $(\mathbb{F}_r,\mathbb{F}_r)$.
In the sequel we will compute the value distribution of $T(a,b)$ for $m$ and $\frac{e(q-1)}{h}$ satisfying some special relationships,  and thus the symbol-pair weight distribution of the code $\C_{(q,m,h,e)}$ is determined from (\ref{eq:cab}).

In order to determine the exponential sum $T(a, b)$ we need to introduce cyclotomic classes, cyclotomic numbers and Gaussian periods.

\subsection{Cyclotomic classes and cyclotomic numbers}

Let $N$ be a non-trivial divisor of $r-1$. Define $C_{i}^{(N,r)}=\alpha^i\langle\alpha^N\rangle$ for $i=0,1,2,\cdots, N-1$, where $\langle\alpha^N\rangle$ denotes
the subgroup of $\mathbb{F}_r^*$ generated by $\alpha^N$. The cosets $C_{i}^{(N,r)}$ are called the cyclotomic classes of order $N$ in $\bF_{r}$. Then we have the following lemmas, which were first given in \cite{CMa2010} and \cite{CDing2010}.
\begin{lemma}\cite{CMa2010}\label{lem:GHW}
Let $i$ be any integer with $0\leq i<e$. We have the following multiset equality:
\begin{equation*}
\begin{aligned}
\left\{xy:y \in \mathbb{F}_q^\ast,x \in C_{i}^{(e,r)}\right\}=\tfrac{q-1}{e}gcd(m,e) \ast C_{i}^{(gcd(m,e),r)},
\end{aligned}
\end{equation*}
where $\tfrac{q-1}{e}gcd(m,e) \ast C_{i}^{(gcd(m,e),r)}$ is the multiset in which each element of $C_{i}^{(gcd(m,e),r)}$ appears with multiplicity $\tfrac{q-1}{e}gcd(m,e)$.
\end{lemma}
\begin{lemma}\cite{CDing2010}\label{lem:GHWq}
Let $q\equiv1 \pmod N$ and $m\equiv0 \pmod N$, then $\bF_{q}^\ast \subset C_{0}^{(N,r)}$.
\end{lemma}

 The cyclotomic numbers of order $N$ are defined by
\[(i,j)^{(N,r)}=\left|(1+C_{i}^{(N,r)})\cap C_{j}^{(N,r)}\right|,\]
where $0\leq i \leq N-1$ and $0\leq j \leq N-1$.  It is very difficult to determine cyclotomic numbers in general. Until now, cyclotomic numbers are known only for a few special cases.
When $N=2$, the cyclotomic numbers are given in the following lemma.
\begin{lemma}\cite{CMa2010}
Let the symbols be given as above. The cyclotomic numbers of order $2$ are given by
\begin{enumerate}
\item[(1)] $(0,0)^{(2,r)}$=$\tfrac{r-5}{4}$\,\,\text{\rm and}\,\,$(0,1)^{(2,r)}=(1,0)^{(2,r)}=(1,1)^{(2,r)}=\tfrac{r-1}{4}$\,\, ${\rm if} \,\,r \equiv1 \pmod 4$.
\item[(2)] $(0,0)^{(2,r)}=(1,0)^{(2,r)}=(1,1)^{(2,r)}=\tfrac{r-3}{4}$\,\,\text{\rm and}\,\,$(0,1)^{(2,r)}=\tfrac{r+1}{4}$\,\, ${\rm if} \,\,r \equiv3 \pmod 4$.
\end{enumerate}
\end{lemma}
In  order to obtain the symbol-pair weight distribution of some cyclic codes, the generalized cyclotomic numbers are needed, which was first given in \cite{aouche2023}.
\begin{definition}\cite{aouche2023}
Let $l$ and $f$ be two non-trivial divisor of $r-1$. We define the generalized cyclotomic numbers of order $(l,f)$ in $\bF_{r}$ by
\[(i,j)^{(l,f,r)}=\left|(1+C_{i}^{(l,r)})\cap C_{j}^{(f,r)}\right|\]
for all $0\leq i\leq l-1$ and $0\leq j \leq f-1$.
\end{definition}
When $l=f$, the generalized cyclotomic numbers are consistent with the cyclotomic numbers. The following lemma gives the fundamental properties of generalized cyclotomic numbers,
which will be used in the proof of the later theorem.
\begin{lemma}\cite{aouche2023}\label{lem:5}
Let $l$ and $f$ be two non-trivial divisors of $r-1$ with $f$ dividing $l$. The cyclotomic numbers of order $(l,f)$ have the following properties.
\begin{enumerate}
\item[(1)] $(i,j)^{(l,f,r)}$=$(-i,j-i)^{(l,f,r)}.$
\item[(2)] $\sum_{i=0}^{l-1}(i,j)^{(l,f,r)}$=$\begin{cases}
(r-1)/f-1, &{\rm if}\ j=0, \\
(r-1)/f,   &{\rm  otherwise.}
\end{cases}$
\item[(3)] $\sum_{j=0}^{{f}-1}(i,j)^{(l,f,r)}$=$\begin{cases}
(r-1)/l-1, &{\rm if}\ -1 \in C_{i}^{(l,r)}, \\
(r-1)/l,   &{\rm  otherwise.}
\end{cases}$
\end{enumerate}
\end{lemma}

\subsection{Gaussian periods}

 Recall that $\chi_1$ and $\chi$ are the canonical additive character of $\bF_{q}$ and $\bF_{r}$, respectively. Then we have
\begin{equation*}
\begin{aligned}
\chi_1(c_0)=e^{2\pi i {\rm Tr}_p^q(c_0)/p}\,\,\,\,\, \text{and}\,\,\,\,\,\chi(c_1)=e^{2\pi i {\rm Tr}_p^r(c_1)/p},
\end{aligned}
\end{equation*}
where $c_0 \in \mathbb{F}_q$ and  $c_1 \in \mathbb{F}_r$, respectively. Hence, $\chi_1$ and $\chi$ are connected by the identity
\begin{equation*}
\begin{aligned}
\chi_1({\rm Tr}_q^r(c_1))=\chi(c_1)  \,\,\,\,\,{\rm for}\ {\rm any}\ c_1 \in \bF_{r}.
\end{aligned}
\end{equation*}
By the definition of canonical characters, it is clear that
\begin{equation*}
\begin{aligned}
\sum_{c_1 \in \bF_{r}}\chi(c_1x)=\begin{cases}
r, &{\rm if}\ x=0,\\
0, &{\rm if}\ x\neq0.
\end{cases}
\end{aligned}
\end{equation*}
Let
 $N$ be a non-trivial divisor of $r-1$. The Gaussian periods are defined to be
\begin{equation*}
\begin{aligned}
\eta_{i}^{(N,r)}=\sum_{x \in C_{i}^{(N,r)}}\chi(x), &&&i=0,1,2,\ldots,N-1.
\end{aligned}
\end{equation*}

In general, the values of the Gaussian period are difficult to calculate, but it can be calculated in some special cases, such as $N=2$.
The following lemma directly follows from \cite[Lemma~4]{CMa2010}.

\begin{lemma}\label{lem:GHW11}
Let $q=p^s$ and $r=q^m$. When $N=2$, the Gaussian periods are given by the following:
\begin{equation}
\begin{aligned}
	 	\eta_{0}^{(2,r)}=\begin{cases}
	 	    \frac{-1-q^{\frac{m}{2}}}{2}, & {\rm if}\,\,p\equiv1\,\,{\rm (mod \,\,4)}\,\,{\rm and}\,\,  sm\equiv 0\,\,{\rm (mod \,\,2)},\,\, {\rm or}\,\, p\equiv3\,\,{\rm (mod \,\,\,4)} \,\,{\rm and}\,\, sm \equiv 0\,\,{\rm (mod \,\,4)},\\
	 	    \frac{q^{\frac{m}{2}}-1}{2}, & {\rm if}\,\,p\equiv1\,\,{\rm (mod \,\,4)}\,\,{\rm and}\,\,  sm\equiv 1\,\,{\rm (mod \,\,2)},\,\, {\rm or}\,\, p\equiv3\,\,{\rm (mod \,\,\,4)} \,\,{\rm and}\,\, sm \equiv 2\,\,{\rm (mod \,\,4)},\\
	 	    \frac{-1-iq^{\frac{m}{2}}}{2}, & {\rm if}\,\,p\equiv3\,\,{\rm (mod \,\,4)}\,\,{\rm and}\,\,  sm\equiv 3\,\,{\rm (mod \,\,4)},\\
	 	    \frac{iq^{\frac{m}{2}}-1}{2}, & {\rm if}\,\,p\equiv3\,\,{\rm (mod \,\,4)}\,\,{\rm and}\,\,  sm\equiv 1\,\,{\rm (mod \,\,4)}
\end{cases}\nonumber
\end{aligned}
\end{equation}
and $\eta_{1}^{(2,r)}=-1-\eta_{0}^{(2,r)},$ where $i^2=-1$.
\end{lemma}

\section{The symbol-pair weight distribution of the first class of cyclic codes}\label{sec3}
In this section, we main consider the symbol-pair weight distribution of
$\mathcal{C}_{(q,m,h,e)}$ for the case that $m>2$ is a positive integer and ${\rm gcd}(m,\tfrac{e(q-1)}{h})=1$, where $q$ is a prime power. In order to obtain the desired result, we first present the following lemma.

\begin{lemma}\label{lem:yy0709}
Let $m>2$, $0\leq i_1,i_2\leq e-1$ and $y_1,y_2\in \bF_{q}^\ast$, then $\frac{\beta^{i_1}(1+\beta gy_1)}{1+gy_1}=\frac{\beta^{i_2}(1+\beta gy_2)}{1+gy_2}$ if and only if $i_1=i_2$ and $y_1=y_2$.
\end{lemma}
\begin{proof}
 If $i_1=i_2$ and $y_1=y_2$, it is clear that $\frac{\beta^{i_1}(1+\beta gy_1)}{1+gy_1}=\frac{\beta^{i_2}(1+\beta gy_2)}{1+gy_2}$. We now prove that  $i_1=i_2$ and $y_1=y_2$ if $\frac{\beta^{i_1}(1+\beta gy_1)}{1+gy_1}=\frac{\beta^{i_2}(1+\beta gy_2)}{1+gy_2}$.

If $i_1\neq i_2$, without loss of generality, we assume that $i_2> i_1$. Let $k=i_2-i_1$, from $\frac{\beta^{i_1}(1+\beta gy_1)}{1+gy_1}=\frac{\beta^{i_2}(1+\beta gy_2)}{1+gy_2}$, we obtain that
\[(1+\beta gy_1)(1+gy_2)=\beta^k(1+\beta gy_2)(1+gy_1),\]
which is the same as
\[ (\beta^k-1)+(\beta^k-\beta)gy_1+(\beta^{k+1}-1)gy_2+(\beta^{k+1}-\beta)g^2y_1y_2=0.\]
Then,
\[(\beta^k-1)(1+\beta g^2y_1y_2)=-\beta(\beta^{k-1}-1)gy_1-(\beta^{k+1}-1)gy_2.\]
Since $\beta \in \bF_{q}$ and $\beta^k \neq 1$ for  $ 0< k \leq e-1$, we have
\begin{equation*}
\tfrac{1+\beta g^2y_1y_2}{g}=\tfrac{-\beta(\beta^{k-1}-1)y_1-(\beta^{k+1}-1)y_2}{(\beta^k-1)} \in \mathbb{F}_q.
\end{equation*}
So, we obtain
\begin{equation}\label{y1y20707}
\left(\frac{1+\beta g^2y_1y_2}{g}\right)^q=(g^{-1}+\beta gy_1y_2)^q=g^{-1}+\beta gy_1y_2.
\end{equation}
It is obvious that
\begin{equation}\label{y1y2}
(g^{-1}+\beta gy_1y_2)^q=g^{-q}+(\beta gy_1y_2)^q=g^{-q}+\beta y_1y_2g^q.
\end{equation}
Combining (\ref{y1y20707}) and (\ref{y1y2}), we have
\begin{equation}\label{y1y20708}
g^{-1}-g^{-q}=\beta y_1y_2(g^q-g).
\end{equation}
It is clear that $g^{q+1}(g^{-1}-g^{-q})=g^q-g$. Since $g \notin \bF_{q}$, from (\ref{y1y20708}) we have $g^{-(q+1)}=\beta y_1y_2 \in \mathbb{F}_q$, i.e., $g^{1-q^2}=1$. By the definition of $g$, we know that
${\rm ord}(g)=\frac{h(r-1)}{q-1},$ then $\frac{h(r-1)}{q-1}\, | \,  q^2-1$, which is contradictory to $r=q^m$ and $m>2$. Hence, we have $i_1=i_2$ if $\frac{\beta^{i_1}(1+\beta gy_1)}{1+gy_1}=\frac{\beta^{i_2}(1+\beta gy_2)}{1+gy_2}$.

If $i_1=i_2$, then $\frac{\beta^{i_1}(1+\beta gy_1)}{1+gy_1}=\frac{\beta^{i_2}(1+\beta gy_2)}{1+gy_2}$ becomes $\frac{1+\beta gy_1}{1+gy_1}=\frac{1+\beta gy_2}{1+gy_2}$,
which is the same as
\[(1+\beta gy_1)(1+gy_2)=(1+gy_1)(1+\beta gy_2).\]
Hence, we have
\[(\beta-1)gy_1=(\beta-1)gy_2,\]
i.e., $y_1=y_2$ since $\beta-1\neq 0$. The desired conclusion then follows.
\end{proof}

\begin{remark}\label{rem:8}
When $m=2$, Lemma \ref{lem:yy0709} does not hold. By the similar way as in Lemma \ref{lem:yy0709}, we can obtain the following result. Let $m=2$, $0\leq i_1,i_2,i_3\leq e-1$ and $y_1,y_2,y_3\in \bF_{q}^\ast$, then $\frac{\beta^{i_1}(1+\beta gy_1)}{1+gy_1}=\frac{\beta^{i_2}(1+\beta gy_2)}{1+gy_2}=\frac{\beta^{i_3}(1+\beta gy_3)}{1+gy_3}$ if and only if $i_1=i_2=i_3$ and $y_1=y_2=y_3$.
\end{remark}

To determine the symbol-pair weight distribution of $\C_{(q,m,h,e)}$, we next deal with the value
distribution of the exponential sum $T(a,b)$, which is given in (\ref{eq:Tab}).
\begin{proposition}\label{pro-01}
Let $m>2$ be a positive integer and ${\rm gcd}(m,\tfrac{e(q-1)}{h})=1$. When $(a,b)$ runs through $(\mathbb{F}_r,\mathbb{F}_r)$, the value distribution of $T(a,b)$ is given as follows:
$$T(a,b)=\left\{\begin{array}{llll}
\frac{h(r-1)}{q-1}, &  1 & {\rm time,}\\
\tfrac{h(q^{m-2}-1)}{q-1},  &(r-1)(r+1-eq) &{\rm times,}\\
\tfrac{h(er+2rq-2r-eq^2)}{eq^2(q-1)}, &e(r-1) &{\rm times,}\\
\tfrac{h(er+rq-r-eq^2)}{eq^2(q-1)},&e(q-1)(r-1) &{\rm times.}
         \end{array}
\right.
$$
\end{proposition}
\begin{proof}
Since $d=\gcd(m,\tfrac{e(q-1)}{h})=1$, we have $C_{(q-1)i/h}^{(d,r)}=\bF_{r}^\ast$ for $0\leq i\leq e-1$. Then $T(a,b)$ in (\ref{eq:Tab}) can be  rewritten as
\begin{equation*}
\begin{aligned}
T(a,b)=\tfrac{h(r-1)}{q^2(q-1)}+\tfrac{2h}{eq^2}\sum_{i=0}^{e-1}\sum_{x \in \bF_{r}^\ast}\chi(x(a+\beta^i b))+\tfrac{h}{eq^2}\sum_{i=0}^{e-1}\sum_{x \in \bF_{r}^\ast}\sum_{y \in \bF_{q}^\ast}\chi(x(a(1+gy)+\beta^i b(1+\beta gy))).
\end{aligned}
\end{equation*}
 Let
\begin{equation}\label{eq:D0708}
D=\left\{\frac{\beta^i(1+\beta gy)}{1+gy}:i=0,1,2,\ldots,e-1, \,\,y \in \bF_{q}^\ast\right\}.
\end{equation}
It is obvious that $\langle\beta\rangle \cap D={\O}$, where $\langle\beta\rangle$ denotes
the subgroup of $\mathbb{F}_r^*$ generated by $\beta$. Otherwise, there exists $\beta^j$ such that
\begin{equation}\label{j0708}
\frac{1+\beta gy}{1+gy}=\beta^j,
\end{equation}
where $0\leq j\leq e-1$. If $j=1$, then (\ref{j0708}) does not hold. If $j\neq 1$, then $\beta^j+\beta^j gy=1+\beta gy$, i.e., $g=\frac{\beta^j-1}{(\beta-\beta^j)y}\in \mathbb{F}_q$ since $\beta \in \mathbb{F}_q$, which is contradictory to $g \notin \mathbb{F}_q$. We now prove this theorem from the following six cases.

\noindent{\bf Case 1:} $a=0$ and $b=0$. In this case, by the definition of $T(a,b)$, we have
\begin{equation*}
\begin{aligned}
T(a,b)=\tfrac{h(r-1)}{q-1}.
\end{aligned}
\end{equation*}

\noindent{\bf Case 2:} $a\neq 0$ and $b=0$. In this case, $T(a,b)$ becomes
\begin{equation*}
\begin{aligned}
T(a,b)&=\tfrac{h(r-1)}{q^2(q-1)}+\tfrac{2h}{eq^2}\sum_{i=0}^{e-1}\sum_{x \in \bF_{r}^\ast}\chi(ax)+\tfrac{h}{eq^2}\sum_{i=0}^{e-1}\sum_{x \in \bF_{r}^\ast}\sum_{y \in \bF_{q}^\ast}\chi(xa(1+gy))\\
&=\tfrac{h(r-1)}{q^2(q-1)}+\tfrac{2h}{eq^2}e(-1)+\tfrac{h}{eq^2}e(q-1)(-1)\\
&=\tfrac{h(q^{m-2}-1)}{q-1}.
\end{aligned}
\end{equation*}
In the second equality, we used the fact that for any $y \in \bF_{q}^\ast$, we have  $1+gy\neq0$.

\noindent{\bf Case 3:} $a=0$ and $b\neq0$. In this case, $T(a,b)$ becomes
\begin{equation*}
\begin{aligned}
T(a,b)&=\tfrac{h(r-1)}{q^2(q-1)}+\tfrac{2h}{eq^2}\sum_{i=0}^{e-1}\sum_{x \in \bF_{r}^\ast}\chi(\beta^i bx)+\tfrac{h}{eq^2}\sum_{i=0}^{e-1}\sum_{x \in \bF_{r}^\ast}\sum_{y \in \bF_{q}^\ast}\chi(x\beta^i b(1+\beta gy)).
\end{aligned}
\end{equation*}
By a discussion similar to Case 2, we have $T(a,b)=\frac{h(q^{m-2}-1)}{q-1}$.

\noindent{\bf Case 4:}  $a \neq 0$, $b \neq 0$, $-\tfrac{a}{b} \notin \langle\beta\rangle$ and $-\tfrac{a}{b} \notin D$, where $D$ is given in (\ref{eq:D0708}).
By an analysis similar to Case 2, we obtain $T(a,b)=\frac{h(q^{m-2}-1)}{q-1}$.

\noindent{\bf Case 5:}  $a \neq 0$, $b\neq 0$ and $-\frac{a}{b}\in D$. From Lemma \ref{lem:yy0709}, we know that for different pairwise $(i,y)$, the values of $\frac{\beta^{i}(1+\beta gy)}{1+gy}$ are different, where $0\leq i\leq e-1$ and $y \in \mathbb{F}_q^\ast$. Then
\begin{equation*}
\begin{aligned}
T(a,b)&=\tfrac{h(r-1)}{q^2(q-1)}+\tfrac{2h}{eq^2}\sum_{i=0}^{e-1}\sum_{x \in \bF_{r}^\ast}\chi(x(a+\beta^i b))+\tfrac{h}{eq^2}\sum_{i=0}^{e-1}\sum_{x \in \bF_{r}^\ast}\sum_{y \in \bF_{q}^\ast}\chi(x(a(1+gy)+\beta^ib(1+\beta gy)))\\
&=\tfrac{h(r-1)}{q^2(q-1)}+\tfrac{2h}{eq^2}e(-1)+\tfrac{h}{eq^2}[r-1+(q-2)(-1)+(e-1)(q-1)(-1)]\\
&=\tfrac{h(er+rq-r-eq^2)}{eq^2(q-1)}.
\end{aligned}
\end{equation*}
In the second equality, we used the fact that $a+\beta^i b\neq0$ for $0\leq i\leq e-1$ since $|\langle\beta\rangle|=e$ and $\langle\beta\rangle \cap D={\O}$.

\noindent{\bf Case 6:} $a \neq 0$, $b \neq 0$ and $-\tfrac{a}{b} \in \langle\beta\rangle$. By a discussion similar to Case 5, we have
\begin{equation*}
\begin{aligned}
T(a,b)&=\tfrac{h(r-1)}{q^2(q-1)}+\tfrac{2h}{eq^2}\sum_{i=0}^{e-1}\sum_{x \in \bF_{r}^\ast}\chi(x(a+\beta^i b))+\tfrac{h}{eq^2}\sum_{i=0}^{e-1}\sum_{x \in \bF_{r}^\ast}\sum_{y \in \bF_{q}^\ast}\chi(x(a(1+gy)+\beta^ib(1+\beta gy)))\\
&=\tfrac{h(r-1)}{q^2(q-1)}+\tfrac{2h}{eq^2}[r-1+(e-1)(-1)]+\tfrac{h}{eq^2}e(q-1)(-1)\\
&=\tfrac{h(er+2rq-2r-eq^2)}{eq^2(q-1)}.
\end{aligned}
\end{equation*}
From all the cases, we obtain that
\begin{equation}
\begin{aligned}
	 	T(a,b)=\begin{cases}
	 	    \tfrac{h(r-1)}{q-1}, &{\rm if}\,\,\, a=0\,\,\,{\rm and}\,\,\, b=0,\\
            \tfrac{h(er+2rq-2r-eq^2)}{eq^2(q-1)}, &{\rm if}\,\,\, a\neq0,\,\,\,b\neq 0\,\,\,{\rm and}\,\,\,-\tfrac{a}{b}\in\langle\beta\rangle,\\
	 		\tfrac{h(er+rq-r-eq^2)}{eq^2(q-1)}, &{\rm if}\ a\neq0,\,\,\,b\neq 0\,\,\,{\rm and}\,\,\,-\tfrac{a}{b}\in D,\\
\tfrac{h(q^{m-2}-1)}{q-1}, &{\rm otherwise.}
	 	\end{cases}\nonumber
\end{aligned}
\end{equation}
From Lemma \ref{lem:yy0709}, we know that $|D|=e(q-1)$. Since $|\langle\beta\rangle|=e$ and $\langle\beta\rangle \cap D={\O}$,
the desired conclusion then follows.
\end{proof}

From above results, one can get the following theorem.

\begin{theorem}\label{thm:1}
Let $\C_{(q,m,h,e)}$ be the cyclic code defined in (\ref{cqmhe}).
Let $m>2$ and $gcd(m,\tfrac{e(q-1)}{h})=1$, then $\C_{(q,m,h,e)}$ is an $\left[\tfrac{h(r-1)}{q-1},2m\right]$ code with the symbol-pair weight enumerator
\begin{equation}
\begin{aligned}
1+e(r-1)z^{hq^{m-2}(eq+e-2)/e}+e(r-1)(q-1)z^{hq^{m-2}(eq+e-1)/e}+(r-1)(r+1-eq)z^{hq^{m-2}(q+1)}.\nonumber
\end{aligned}
\end{equation}
In addition, the code achieves the largest minimum symbol-pair distance compared with its minimum Hamming distance if $e=2$,  i.e., $d_p(\C_{(q,m,h,e)})=2d_H(\C_{(q,m,h,e)})=hq^{m-1}$.
\end{theorem}

\begin{proof}
From Proposition \ref{pro-01}, we know that the value of $T(a,b)$ is less than the length of $\C_{(q,m,h,e)}$ for any pair $(a,b) \neq (0,0)$.
Hence, the dimension of $\C_{(q,m,h,e)}$ is $2m$. Combining Proposition \ref{pro-01} and (\ref{eq:cab}),
 the symbol-pair weight enumerator of $\C_{(q,m,h,e)}$ is determined. The authors in \cite[Theorem 5]{CMa2010} showed that $d_H(\C_{(q,m,h,e)})=\frac{hq^{m-1}(e-1)}{e}$. It is clear that when $e=2$,  the code $\C_{(q,m,h,e)}$
satisfy $d_p(\C_{(q,m,h,e)})=2d_H(\C_{(q,m,h,e)})$, which implies that the code achieves the highest minimum symbol-pair distance compared with its minimum Hamming distance since $d_p(\C) \leq 2d_H(\C)$ for any code. The desired result then follows.
\end{proof}

\begin{example}\label{ex1}
Let $q=3$, $m=3$, $e=2$ and $h=2$.  Then $\C_{(q,m,h,e)}$ is a $[26,6]$ cyclic code over $\bF_{3}$ with the symbol-pair enumerator $1+52z^{18}+104z^{21}+572z^{24}$. The result is verified by Magma programs.
\end{example}

\begin{example}\label{ex2}
Let $q=4$, $m=4$, $e=3$ and $h=3$.  Then $\C_{(q,m,h,e)}$ is a $[255,8]$ cyclic code over $\bF_{4}$ with the symbol-pair enumerator $1+765z^{208}+2295z^{224}+62475z^{240}$. The result is verified by Magma programs.
\end{example}

\begin{example}\label{ex2}
Let $q=5$, $m=3$, $e=4$ and $h=4$.  Then $\C_{(q,m,h,e)}$ is a $[124,6]$ cyclic code over $\bF_{5}$ with the symbol-pair enumerator $1+496z^{110}+1984z^{115}+13144z^{120}$. The result is verified by Magma programs.
\end{example}

\begin{remark}
When $m=1$,  it is easy to get that the code $\C_{(q,m,h,e)}$ defined in (\ref{cqmhe}) is an $[h,2,h]$ constant symbol-pair weight code. By the definition of Singleton type bound, we know that $\C_{(q,m,h,e)}$ is an MDS symbol-pair code. When $m=2$, by a method similar to Proposition \ref{pro-01}, we can derive that $\C_{(q,m,h,e)}$ defined in (\ref{cqmhe}) is a three symbol-pair weight code from Remark \ref{rem:8}.
\end{remark}

\section{The pair weight distribution of the second class of cyclic codes}\label{sec3}

In this section, let $e=2$ and $q$ be an odd prime power. We consider the possible symbol-pair weights of
$\mathcal{C}_{(q,m,h,2)}$ for the case that ${\rm gcd}(m,\tfrac{2(q-1)}{h})=2$.
Moreover, when $m=2$, we show that $\mathcal{C}_{(q,2,h,2)}$ is a three symbol-pair weight code and determine its symbol-pair weight distribution.
\subsection{The case  $\tfrac{q-1}{h} \equiv 1 \pmod 2$}
In this subsection, we always assume that $\tfrac{q-1}{h} \equiv 1 \pmod 2$. Firstly, we give the possible symbol-pair weights of
$\mathcal{C}_{(q,m,h,2)}$.
\begin{theorem}\label{thm:2}
Let ${\rm gcd}(m,\tfrac{2(q-1)}{h})=2$. When $(a,b)$ runs through $(\mathbb{F}_r,\mathbb{F}_r)$, then the set of all possible symbol-pair weights of $\C_{(q,m,h,2)}$ is
\begin{equation}\label{eq:0819}
\begin{split}
&\Big\{ 0, \, hq^{m-2}(q+1),\, hq^{m-1} \pm hq^{\tfrac{m}{2}-1},\,hq^{m-2}(q+\tfrac{1}{2})\pm hq^{\tfrac{m}{2}-2}\big(\tfrac{3}{2}+\big(\tfrac{2(q-1)}{h},0\big)^{\big(\tfrac{r-1}{q-1},2,r\big)}\big),\\
&hq^{m-2}(q+1)+hq^{\tfrac{m}{2}-2}(t-l+2\varepsilon)\Big\},
\end{split}
\end{equation}
where $0\leq l,t\leq q-1$ and $\varepsilon=0, \pm 1$. Moreover, the code achieves the largest minimum symbol-pair distance compared with its minimum Hamming distance, i.e.,
$$d_p(\C_{(q,m,h,2)})=2d_H(\C_{(q,m,h,2)})=hq^{m-1} - hq^{\tfrac{m}{2}-1}.$$
\end{theorem}

\begin{proof}
In order to obtain the desired results, from (\ref{eq:cab}) we only need to consider the possible values of $T(a,b)$ for $(a,b)$ running through $(\mathbb{F}_r,\mathbb{F}_r)$.

 By the definition of cyclotomic class, it is clear that $C_{(q-1)/h}^{(2,r)}=C_{1}^{(2,r)}$ if $\tfrac{q-1}{h} \equiv 1\pmod 2$. Since $e=2$, by the definition of $\beta$, we have $\beta=-1$. Then from (\ref{eq:Tab}) we have
\begin{equation}\label{Tab0710}
\begin{aligned}
T(a,b)=&\tfrac{h(r-1)}{q^2(q-1)}+\tfrac{2h}{q^2}\Big[\sum_{x \in C_{0}^{(2,r)}}\chi(x(a+b))+\sum_{x \in C_{1}^{(2,r)}}\chi(x(a-b))\Big]\\
&+\tfrac{h}{q^2}\Big[\sum_{x \in C_{0}^{(2,r)}}\sum_{y \in \bF_{q}^\ast}\chi(x(a+b+(a- b)gy))+\sum_{x \in C_{1}^{(2,r)}}\sum_{y \in \bF_{q}^\ast}\chi(x(a-b+(a+b
)gy))\Big].\\
\end{aligned}
\end{equation}
We now prove this theorem from the following eight cases.

\noindent{\bf Case 1:} $a=0$ and $b=0$. By the definition of $T(a,b)$, we have
\begin{equation*}
\begin{aligned}
T(a,b)=\tfrac{h(r-1)}{q-1}.
\end{aligned}
\end{equation*}

\noindent{\bf Case 2:} $a\neq0$ and $b=0$. It is easy to see that (\ref{Tab0710}) can be written as
\begin{equation*}
\begin{aligned}
T(a,b)=&\tfrac{h(r-1)}{q^2(q-1)}+\tfrac{2h}{q^2}\Big[\sum_{x \in C_{0}^{(2,r)}}\chi(ax)+\sum_{x \in C_{1}^{(2,r)}}\chi(ax)\Big]\\
&+\tfrac{h}{q^2}\Big[\sum_{x \in C_{0}^{(2,r)}}\sum_{y \in \bF_{q}^\ast}\chi(xa(1+gy))+\sum_{x \in C_{1}^{(2,r)}}\sum_{y \in \bF_{q}^\ast}\chi(xa(1+gy))\Big].
\end{aligned}
\end{equation*}
Since $g\notin\bF_{q}$, for all $y\in\bF_{q}^\ast$ we have $1+gy\neq0$. Then
\begin{equation*}
\begin{aligned}
T(a,b)&=\tfrac{h(r-1)}{q^2(q-1)}+\tfrac{2h}{q^2}\sum_{x\in F_r^\ast}\chi(ax)+\tfrac{h}{q^2}\sum_{x\in F_r^\ast}\sum_{y \in \bF_{q}^\ast}\chi(xa(1+gy))\\
&=\tfrac{h(r-1)}{q^2(q-1)}+\tfrac{2h}{q^2}(-1)+\tfrac{h}{q^2}(q-1)(-1)\\
&=\tfrac{h(q^{m-2}-1)}{q-1}.
\end{aligned}
\end{equation*}

\noindent{\bf Case 3:} $a=0$ and $b\neq0$. By the same way as in Case 2, we have $T(a,b)=\frac{h(q^{m-2}-1)}{q-1}$.

\noindent{\bf Case 4:} $a=-b\neq0$. Obviously, (\ref{Tab0710}) can be written as
\begin{equation*}
\begin{aligned}
T(a,b)=&\tfrac{h(r-1)}{q^2(q-1)}+\tfrac{2h}{q^2}\Big[\sum_{x \in C_{0}^{(2,r)}}\chi(0)+\sum_{x \in C_{1}^{(2,r)}}\chi(2ax)\Big]\\
&+\tfrac{h}{q^2}\Big[\sum_{x \in C_{0}^{(2,r)}}\sum_{y \in \bF_{q}^\ast}\chi(2axgy)+\sum_{x \in C_{1}^{(2,r)}}\sum_{y \in \bF_{q}^\ast}\chi(2ax)\Big].
\end{aligned}
\end{equation*}
Since ${\rm gcd}(m,\tfrac{2(q-1)}{h})=2$, we have that $m$ must be even. From Lemma \ref{lem:GHWq}, we know that $y \in C_{0}^{(2,r)}$.
If $a \in C_0^{(2,r)}$, then
 \begin{equation}
\begin{aligned}
T(a,b)&=\tfrac{h(r-1)}{q^2(q-1)}+\tfrac{2h}{q^2}\Big[\tfrac{r-1}{2}+\eta_1^{(2,r)}\Big]+\tfrac{h}{q^2}\Big[(q-1)\eta_1^{(2,r)}+(q-1)\eta_1^{(2,r)}\Big]\\
&=\tfrac{h(r-1)}{q^2(q-1)}+\tfrac{h(r-1)}{q^2}+\tfrac{2h}{q^2}\eta_1^{(2,r)}+\tfrac{2h}{q^2}(q-1)\eta_1^{(2,r)}\\
&=\tfrac{h}{q}(\tfrac{r-1}{q-1}+2\eta_1^{(2,r)})\nonumber
\end{aligned}
\end{equation}
since $g=\alpha^{\tfrac{q-1}{h}} \in C_1^{(2,r)}$. If $a \in C_1^{(2,r)}$, by the same way as above, we have $T(a,b)=\tfrac{h}{q}(\tfrac{r-1}{q-1}+2\eta_0^{(2,r)})$.

\noindent{\bf Case 5:} $a=b\neq0$. By the same way as in Case 4, we have
\begin{equation*}
\begin{aligned}
T(a,b)=\begin{cases}
\frac{h}{q}(\frac{r-1}{q-1}+2\eta_0^{(2,r)}), &{\rm if}\,\, a \in C_0^{(2,r)},\\
\frac{h}{q}(\frac{r-1}{q-1}+2\eta_1^{(2,r)}), &{\rm if}\,\, a \in C_1^{(2,r)}.
\end{cases}
\end{aligned}
\end{equation*}

\noindent{\bf Case 6:} $\frac{a}{b} \in B$, where
\begin{equation}\label{eq:B1B2}
\begin{split}
B=\left\{\tfrac{gy-1}{1+gy}:y \in \bF_{q}^\ast\right\}.
\end{split}
\end{equation}
Since $\frac{a}{b} \in B$, there exists $y_1 \in \mathbb{F}_q^*$ such that $\frac{gy_1-1}{1+gy_1}=\frac{a}{b}$, i.e.,
\begin{equation*}
a+b=(b-a)gy_1.
\end{equation*}
Let $y\in \mathbb{F}_q^*$, then
\begin{equation}\label{eq:0712}
a+b+(a-b)gy=(b-a)gy_1+(a-b)gy=(a-b)g(y-y_1)
\end{equation}
and
\begin{equation}\label{eq:07111}
a-b+(a+b)gy=(a-b)-(a-b)g^2yy_1=(a-b)(1-g^2yy_1).
\end{equation}
From Lemma \ref{lem:GHWq},
for any $y \in \bF_{q}^\ast\backslash\left\{y_1\right\}$ we have $y-y_1 \in C_{0}^{(2,r)}$. Since $\tfrac{q-1}{h} \equiv 1\pmod 2$, we obtain
$g=\alpha^{\tfrac{q-1}{h}}\in C_{1}^{(2,r)}$.
 Then
\begin{equation}\label{eq:qhl}g(y-y_1)\in C_{1}^{(2,r)}.\end{equation}
Moreover, it is clear that
\begin{equation}\label{eq:07113}
\{-g^2yy_1:y\in \mathbb{F}_q^*\}=\{-g^2y:y\in \mathbb{F}_q^*\}=-\alpha^{\frac{2(q-1)}{h}}\langle\alpha^{\frac{r-1}{q-1}}\rangle=\alpha^{\frac{2(q-1)}{h}}\langle\alpha^{\frac{r-1}{q-1}}\rangle=C_{\frac{2(q-1)}{h}}^{(\frac{r-1}{q-1},r)}.
\end{equation}

If $a-b\in C_{1}^{(2,r)}$, from (\ref{eq:0712})-(\ref{eq:07113}) and the definition of generalized cyclotomic numbers, we obtain that
$$\left\{\begin{array}{llll}
&\left|\{(a-b)(1-g^2y)\,:\,y\in \mathbb{F}_q^*\} \cap C_{0}^{(2,r)}\right|=(\tfrac{2(q-1)}{h},1)^{(\tfrac{r-1}{q-1},2,r)},\\
&\left|\{(a-b)(1-g^2y)\,:\,y\in \mathbb{F}_q^*\} \cap C_{1}^{(2,r)}\right|=(\tfrac{2(q-1)}{h},0)^{(\tfrac{r-1}{q-1},2,r)},\\
&a+b+(a-b)gy\in C_{0}^{(2,r)}\,\, \text{for any}\,\, y \in \bF_{q}^\ast\backslash\left\{y_1\right\}.
         \end{array}
\right.$$
Hence, from (\ref{Tab0710}) we have
\begin{equation}\label{eq:081901}
\begin{aligned}
T(a,b)=&\tfrac{h(r-1)}{q^2(q-1)}+\tfrac{2h}{q^2}\Big[\eta_{0}^{(2,r)}+\eta_{0}^{(2,r)}\Big]+\tfrac{h}{q^2}\Big[\tfrac{r-1}{2}+(q-2)\eta_{0}^{(2,r)}\\
&+(\tfrac{2(q-1)}{h},0)^{(\tfrac{r-1}{q-1},2,r)}\eta_{0}^{(2,r)}+(\tfrac{2(q-1)}{h},1)^{(\tfrac{r-1}{q-1},2,r)}\eta_{1}^{(2,r)}\Big]\\
=&\tfrac{h(r-1)}{q^2(q-1)}+\tfrac{h(r-1)}{2q^2}-\tfrac{h}{q^2}(\tfrac{2(q-1)}{h},1)^{(\tfrac{r-1}{q-1},2,r)}\\
&+\tfrac{h}{q^2}\eta_{0}^{(2,r)}\Big[q+2+(\tfrac{2(q-1)}{h},0)^{(\tfrac{r-1}{q-1},2,r)}-(\tfrac{2(q-1)}{h},1)^{(\tfrac{r-1}{q-1},2,r)}\Big]\\
=&\tfrac{h(q^{m-2}-1)}{q-1}+\tfrac{h}{q^2}(\tfrac{q^m}{2}+\tfrac{3}{2}+3\eta_{0}^{(2,r)}+(\tfrac{2(q-1)}{h},0)^{(\tfrac{r-1}{q-1},2,r)}(1+2\eta_{0}^{(2,r)})).
\end{aligned}
\end{equation}
Similarly, we obtain
$$T(a,b)=\tfrac{h(q^{m-2}-1)}{q-1}+\tfrac{h}{q^2}(\tfrac{q^m}{2}+\tfrac{3}{2}+3\eta_{1}^{(2,r)}+(\tfrac{2(q-1)}{h},0)^{(\tfrac{r-1}{q-1},2,r)}(1+2\eta_{1}^{(2,r)}))$$
if $a-b\in C_{0}^{(2,r)}$.

\noindent{\bf Case 7:} $-\frac{a}{b} \in B$, where $B$ is defined in (\ref{eq:B1B2}). With an analysis similar to Case 6, we have
\begin{equation*}
\begin{aligned}
T(a,b)
&=\tfrac{h(q^{m-2}-1)}{q-1}+\tfrac{h}{q^2}(\tfrac{q^m}{2}+\tfrac{3}{2}+3\eta_{0}^{(2,r)}+(\tfrac{2(q-1)}{h},0)^{(\tfrac{r-1}{q-1},2,r)}(1+2\eta_{0}^{(2,r)}))
\end{aligned}
\end{equation*}
if $a-b\in C_{1}^{(2,r)}$, and
$$T(a,b)=\tfrac{h(q^{m-2}-1)}{q-1}+\tfrac{h}{q^2}(\tfrac{q^m}{2}+\tfrac{3}{2}+3\eta_{1}^{(2,r)}+(\tfrac{2(q-1)}{h},0)^{(\tfrac{r-1}{q-1},2,r)}(1+2\eta_{1}^{(2,r)}))$$
if $a{-b}\in C_{0}^{(2,r)}$.

\noindent{\bf Case 8:} $a\neq\pm b\neq 0$ and $\pm\frac{a}{b} \notin B$. By the definition of $B$, for any $y\in \bF_{q}^\ast$, we have $a+b+(a-b)gy \neq0$ and $a-b+(a+b)gy\neq0$.
Let
\begin{equation}\label{eq:072701}
t=|\{a+b+(a-b)gy\,:\,y\in \mathbb{F}_q^*\} \cap C_{0}^{(2,r)}|\,\,\, \text{and}\,\,\,l=|\{a-b+(a+b)gy\,:\,y\in \mathbb{F}_q^*\} \cap C_{0}^{(2,r)}|.
\end{equation}
 From (\ref{eq:072701}) we know that
$$|\{a+b+(a-b)gy\,:\,y\in \mathbb{F}_q^*\} \cap C_{1}^{(2,r)}|=q-1-t\,\,\, \text{and}\,\,\,|\{a-b+(a+b)gy\,:\,y\in \mathbb{F}_q^*\} \cap C_{1}^{(2,r)}|=q-1-l.$$
If $a+b\in C_{0}^{(2,r)}$ and $a-b\in C_{0}^{(2,r)}$,
from (\ref{Tab0710}) we have
\begin{equation*}
\begin{aligned}
T(a,b)&=\tfrac{h(r-1)}{q^2(q-1)}+\tfrac{2h}{q^2}\Big[\eta_{0}^{(2,r)}+\eta_{1}^{(2,r)}\Big]+\tfrac{h}{q^2}\Big[t\eta_{0}^{(2,r)}+(q-1-t)\eta_{1}^{(2,r)}
+l\eta_{1}^{(2,r)}+(q-1-l)\eta_{0}^{(2,r)}\Big]\\
&=\tfrac{h(r-1)}{q^2(q-1)}-\tfrac{2h}{q^2}+\tfrac{h}{q^2}[2(t-l)\eta_{0}^{(2,r)}-q+1+t-l]\\
&=\tfrac{h(r-1)}{q^2(q-1)}-\tfrac{2h}{q^2}+\tfrac{2h}{q^2}(t-l)\eta_{0}^{(2,r)}-\tfrac{h(q-1)}{q^2}+\tfrac{h(t-l)}{q^2}\\
&=\tfrac{h(r-1)}{q^2(q-1)}-\tfrac{2h}{q^2}-\tfrac{h(q-1)}{q^2}+\tfrac{h(t-l)}{q^2}(1+2\eta_{0}^{(2,r)})\\
&=\tfrac{h(q^{m-2}-1)}{q-1}+\tfrac{h}{q^2}(t-l)(1+2\eta_{0}^{(2,r)}).
\end{aligned}
\end{equation*}
In the second equality, we used the fact that $\eta_{1}^{(2,r)}+\eta_{0}^{(2,r)}=-1$, which is given in Lemma \ref{lem:GHW11}.
Similarly, we obtain
\begin{equation*}
\begin{aligned}
T(a,b)=\begin{cases}
\tfrac{h(q^{m-2}-1)}{q-1}+\tfrac{h}{q^2}(t-l)(1+2\eta_{0}^{(2,r)}), &{\rm if}\,\, a+b\in C_{1}^{(2,r)} \,\, \text{and} \,\, a-b\in C_{1}^{(2,r)},\\
\tfrac{h(q^{m-2}-1)}{q-1}+\tfrac{h}{q^2}(2+4\eta_{0}^{(2,r)})+\tfrac{h}{q^2}(t-l)(1+2\eta_{0}^{(2,r)}), &{\rm if}\,\, a+b\in C_{0}^{(2,r)}  \,\, \text{and} \,\,a-b\in C_{1}^{(2,r)},\\
\tfrac{h(q^{m-2}-1)}{q-1}-\tfrac{h}{q^2}(2+4\eta_{0}^{(2,r)})+\tfrac{h}{q^2}(t-l)(1+2\eta_{0}^{(2,r)}),  &{\rm if}\,\, a+b\in C_{1}^{(2,r)}  \,\, \text{and} \,\,a-b\in C_{0}^{(2,r)},
\end{cases}
\end{aligned}
\end{equation*}
where $t,l$ are given in (\ref{eq:072701}).

Let $E_i=\tfrac{q^m}{2}+\tfrac{3}{2}+3\eta_{i}^{(2,r)}+(\tfrac{2(q-1)}{h},0)^{(\tfrac{r-1}{q-1},2,r)}(1+2\eta_{i}^{(2,r)})$ for  $i=0$ or $i=1$, and $E_2=(t-l)(1+2\eta_{0}^{(2,r)})$.
Combining all the cases, we have
\begin{equation}\label{eq:082501141}
T(a,b)=\tfrac{h(q^{m-2}-1)}{q-1}+\tfrac{h}{q^2}T'(a,b),
\end{equation}
where
\begin{equation*}
\begin{aligned}
T'(a,b)=\begin{cases}
q^{m}(q+1),  & {\rm if }\,\,\,a=0\,\,\,{\rm and}\,\,\,b=0,\\
0,   & {\rm if} \,\,\, a\neq0\,\,\,{\rm and}\,\,\,b=0,\,\,\, {\rm or}\,\,\,
a=0\,\,\,{\rm and}\,\,\,b\neq0,
\\
q^{m}+q+2q\eta_{1}^{(2,r)},  &{\rm if}\,\,\, a=-b\in C_0^{(2,r)},\,\,\, {\rm or}\,\,\,  a=b\in C_1^{(2,r)},\\
q^{m}+q+2q\eta_{0}^{(2,r)},   &{\rm if}\,\,\, a=-b\in C_1^{(2,r)},\,\,\, {\rm or}\,\,\,  a=b\in C_0^{(2,r)},\\
E_0, &{\rm if} \,\,\,  \tfrac{a}{b}\in B\,\,\,{\rm and}\,\,\,a-b\in C_1^{(2,r)},{\rm or}\,-\tfrac{a}{b}\in B\,\,\,{\rm and}\,\,\,a-b\in C_1^{(2,r)},\\
E_1, &{\rm if} \,\,\,   \tfrac{a}{b}\in B\,\,\,{\rm and}\,\,\,a-b\in C_0^{(2,r)},{\rm or}\,-\tfrac{a}{b}\in B\,\,\,{\rm and}\,\,\,a-b\in C_0^{(2,r)},\\
E_2,  & {\rm if}\,\,\,a\neq0,\,\,\, \pm\tfrac{a}{b}\notin B,\,\,\,  a+b\in C_i^{(2,r)} \,\,\,{\rm and} \,\,\,a-b\in C_i^{(2,r)},\\
E_2+2+4\eta_{0}^{(2,r)}, & {\rm if}\,\,\, \pm\tfrac{a}{b}\notin B,\,\,\,  a+b\in C_0^{(2,r)}\,\,\,{\rm and}\,\,\,a-b\in C_1^{(2,r)},\\
E_2-2-4\eta_{0}^{(2,r)}, &{\rm if}\,\,\,\pm\tfrac{a}{b}\notin B,\,\,\,  a+b\in C_1^{(2,r)}\,\,\,{\rm and} \,\,\,a-b\in C_0^{(2,r)}
\end{cases}
\end{aligned}
\end{equation*}
for $i=0$ or $i=1$. Obviously, we have $0\leq t,l\leq q-1$ by definition. Then from Lemma \ref{lem:GHW11} and (\ref{eq:cab}), the possible symbol-pair weights of $C_{(q,m,h,2)}$ can be given and we show this result in (\ref{eq:0819}).

From (\ref{eq:0819}), it is obvious that the minimum possible symbol-pair weight of $\C_{(q,m,h,2)}$ is $hq^{m-1} - hq^{\tfrac{m}{2}-1}$. In \cite[Theorem 6]{CMa2010}, the authors showed that $d_H(\C_{(q,m,h,2)})=\frac{1}{2}(hq^{m-1} - hq^{\tfrac{m}{2}-1})$, then $d_p(\C_{(q,m,h,2)})=hq^{m-1} - hq^{\tfrac{m}{2}-1}$ since $ d_p(\C) \leq 2d_H(\C)$ for any code $\C$. The desired result then follows.
\end{proof}

 \begin{remark}
By Magma programs, all the possible symbol-pair weights of $\mathcal{C}_{(q,m,h,2)}$ in (\ref{eq:0819}) will appear for many different $m$ and $q$. For example,
let $q=3$, $m=6$, $e=2$ and $h=2$, then $\C_{(3,6,2,2)}$ is a $[728,12]$ code over  $\mathbb{F}_3$ with the symbol-pair weights
$$\left\{468,504,558,576,624,630,636,642,648,654,660,666,672\right\}.$$
\end{remark}

In general, it is very hard to determine the symbol-pair weight distribution of $\mathcal{C}_{(q,m,h,2)}$ in Theorem \ref{thm:2}.
For the special case $m=2$,
 we can show that $\mathcal{C}_{(q,2,h,2)}$ is a three symbol-pair weight code and determine its symbol-pair weight distribution. For the convenience of narration, we first give some  preliminary lemmas.

\begin{lemma}\label{lem:GHWs720}
Let $T(a,b)$ be given in (\ref{eq:Tab}) and $r=q^2$, then
$$T(a,b)=\left\{\begin{array}{llll}
0,  & r-1 & {\rm times,}\\
2h, & r-1 &{\rm times}
         \end{array}
\right.
$$
when $(a,b)$ runs over $\big\{ (a,b) \in (\mathbb{F}_r, \mathbb{F}_r)\,:\,a=b\neq 0,\,\,\,{\rm or}\,\,\,a=-b\neq 0\}$.
\end{lemma}
\begin{proof}  From Case 4 in Theorem \ref{thm:2}, we have
\begin{equation*}
\begin{aligned}
T(a,b)=\begin{cases}
\frac{h}{q}(q+1+2\eta_1^{(2,r)}), &{\rm if}\,\, a=-b \in C_0^{(2,r)},\\
\frac{h}{q}(q+1+2\eta_0^{(2,r)}), &{\rm if}\,\, a=-b \in C_1^{(2,r)}.
\end{cases}
\end{aligned}
\end{equation*}
Then from  Lemma \ref{lem:GHW11}, we know that $\eta_0^{(2,r)}=\tfrac{-1-q}{2}$ or $\eta_0^{(2,r)}=\tfrac{-1+q}{2}$. Hence,
\begin{equation}\label{eq:Tab823}
\begin{aligned}
T(a,b)=\left\{\begin{array}{llll}
0,  & \frac{r-1}{2} & {\rm times,}\\
2h,  & \frac{r-1}{2}  &{\rm times}
         \end{array}
\right.
\end{aligned}
\end{equation}
when $(a,b)$ runs through $\big\{ (a,b) \in (\mathbb{F}_r, \mathbb{F}_r)\,:\, a=-b\neq 0\}$. From Case 5 in Theorem \ref{thm:2}, we have the same result as (\ref{eq:Tab823}) when $(a,b)$ runs through $\big\{ (a,b) \in (\mathbb{F}_r, \mathbb{F}_r)\,:\, a=b\neq 0\}$.
The conclusion then follows.
\end{proof}

\begin{lemma}\label{lem:GHWt}
Let
$H_{i,0}=\big\{(a,b) \in (\mathbb{F}_r, \mathbb{F}_r)\,:\, \frac{(-1)^i a}{b} \in B\,\,\text{and}\,\, a-b \in C_0^{(2,r)}\big\}$
and
$H_{i,1}=\big\{(a,b) \in (\mathbb{F}_r, \mathbb{F}_r)\,:\, \frac{(-1)^i a}{b} \in B\,\,\text{and}\,\, a-b \in C_1^{(2,r)}\big\},$
where $i=0,1$ and $B$ is given in (\ref{eq:B1B2}). Then $|H_{(i,0)}|=|H_{(i,1)}|=\frac{(r-1)(q-1)}{2}$.
\end{lemma}
\begin{proof}
 We only prove the value of $|H_{(0,0)}|$,  the values of $H_{(1,0)}$, $H_{(0,1)}$ and $H_{(1,1)}$ can be shown similarly.

Since $\frac{a}{b} \in B$, by the definition of $B$ we have $a=\frac{gy-1}{1+gy}b$. Then $a-b=\frac{-2b}{1+gy}$. By Lemma \ref{lem:GHWq}, we know that $y \in C_0^{(2,r)}$ for any $y \in \mathbb{F}_q^*$. Hence,
$(a,b) \in H_{(0,0)}$ if and only if $(b,y) \in L_{(0,0)}\bigcup L_{(1,1)}$, where
$$L_{(0,0)}=\left\{{(b,y) \in (\mathbb{F}_r, \mathbb{F}_q^*)\,:\, b \in C_0^{(2,r)}\,\,\text{and}\,\, 1+gy \in C_{0}^{(2,r)}}\right\}$$
and
$$L_{(1,1)}=\left\{{(b,y) \in (\mathbb{F}_r, \mathbb{F}_q^*)\,:\, b \in C_1^{(2,r)}\,\,\text{and}\,\, 1+gy \in C_{1}^{(2,r)}}\right\}.$$
Since $$1+gy=1+\alpha^{\frac{q-1}{h}}\langle\alpha^{\frac{r-1}{q-1}}\rangle=1+C_{\frac{q-1}{h}}^{(\frac{r-1}{q-1},r)},$$ by the definition of generalized cyclotomic numbers, we have
\begin{equation}
\begin{split}
|L_{(0,0)}|&=\left|\{b\in \mathbb{F}_r\, :\,  b \in C_0^{(2,r)}\}\right|\left|\{y\in \mathbb{F}_q^*\, :\, 1+gy \in C_{0}^{(2,r)}\}\right|\\
&=\frac{r-1}{2}\left|(1+C_{\frac{q-1}{h}}^{(\frac{r-1}{q-1},r)})\cap C_{0}^{(2,r)}\right|\\
&=\tfrac{r-1}{2}(\tfrac{q-1}{h},0)^{(\tfrac{r-1}{q-1},2,r)}.
\end{split}
\end{equation}
Similarly, we obtain that
$|L_{(1,1)}|=\tfrac{r-1}{2}(\tfrac{q-1}{h},1)^{(\tfrac{r-1}{q-1},2,r)}.$
It is easy to see that $L_{(0,0)}\cap L_{(1,1)}={\O}$. Hence,
\begin{equation}
|H_{(0,0)}|=|L_{(0,0)}|+|L_{(1,1)}|
=\tfrac{r-1}{2}((\tfrac{q-1}{h},0)^{(\tfrac{r-1}{q-1},2,r)}+(\tfrac{q-1}{h},1)^{(\tfrac{r-1}{q-1},2,r)})=\tfrac{(r-1)(q-1)}{2}.
\end{equation}
The desired result then follows.
\end{proof}

\begin{lemma}\label{lem:GHWs7201}
Let $B$ be given in (\ref{eq:B1B2}) and $r=q^2$. Then
$$T(a,b)=\left\{\begin{array}{llll}
0,  & (r-1)(q-1) & {\rm times,}\\
h,  & (r-1)(q-1) &{\rm times}
         \end{array}
\right.
$$
 when $(a,b)$ runs over $\big\{ (a,b) \in (\mathbb{F}_r, \mathbb{F}_r)\,:\, \frac{a}{b} \in B\,\,\, {\rm or}\,\,\,-\frac{a}{b} \in B\big\}$, where $T(a,b)$ is given in (\ref{eq:Tab}).
\end{lemma}
\begin{proof}
Since $r=q^2$, by the definition of $T(a,b)$, we know that $0\leq T(a,b)\leq h(q+1)$.
 We now prove the value distribution of $T(a,b)$ for  $(a,b)$ running through $\big\{ (a,b) \in (\mathbb{F}_r, \mathbb{F}_r)\,:\, \frac{a}{b} \in B\big\}$ from the following two cases.

\noindent{\bf Case 1:} $a-b\in C_{1}^{(2,r)}$.
In this case, from (\ref{eq:081901}) we know that
\begin{equation}\label{eq:0714}
T(a,b)=\tfrac{h}{q^2}(\tfrac{q^2}{2}+\tfrac{3}{2}+3\eta_{0}^{(2,r)}+(\tfrac{2(q-1)}{h},0)^{(q+1,2,r)}(1+2\eta_{0}^{(2,r)})).
\end{equation}
If $p\equiv1\pmod 4$, or $p\equiv3\pmod 4$ and $s$ is even, we have $\eta_{0}^{(2,r)}=\tfrac{-1-q}{2}$ by Lemma \ref{lem:GHW11}. Then from (\ref{eq:0714}), $T(a,b)$ can be rewritten as
\begin{equation*}
T(a,b)=\tfrac{h}{q}(\tfrac{q-3}{2}-(\tfrac{2(q-1)}{h},0)^{(q+1,2,r)}).
\end{equation*}
Since $0\leq T(a,b)\leq h(q+1)$ and $\gcd(q,h)=1$, we obtain $$\tfrac{q-3}{2}-(\tfrac{2(q-1)}{h},0)^{(q+1,2,r)}=kq,$$
which means that
\begin{equation}\label{eq:071422}
(\tfrac{2(q-1)}{h},0)^{(q+1,2,r)}=\tfrac{q-3}{2}-kq=\tfrac{-3+(1-2k)q}{2}=-\tfrac{(2k-1)q+3}{2},
\end{equation}
where $0\leq k \leq q+1$. By the definition of the generalized cyclotomic numbers, we know that $0 \leq (\tfrac{2(q-1)}{h},0)^{(q+1,2,r)} \leq q-1$. Then (\ref{eq:071422}) holds if and only if $k=0$. Hence, we obtain that $T(a,b)=0$.

If $p\equiv3\pmod 4$ and $s$ is odd, we have $\eta_{0}^{(2,r)}=\tfrac{q-1}{2}$ by Lemma \ref{lem:GHW11}. Then from (\ref{eq:0714}) we obtain
\begin{equation*}
\begin{aligned}
T(a,b)=\tfrac{h}{q}(\tfrac{q+3}{2}+(\tfrac{2(q-1)}{h},0)^{(q+1,2,r)}).
\end{aligned}
\end{equation*}
With an analysis similar as above, then
\begin{equation*}
(\tfrac{2(q-1)}{h},0)^{(q+1,2,r)}=kq-\tfrac{q+3}{2}=\tfrac{(2k-1)q-3}{2}.
\end{equation*}
Since $0 \leq (\tfrac{2(q-1)}{h},0)^{(q+1,2,r)} \leq q-1$, we have
$(\tfrac{2(q-1)}{h},0)^{(q+1,2,r)}=\tfrac{q-3}{2}$ and $T(a,b)=h$.

\noindent{\bf Case 2:} $a-b\in C_{0}^{(2,r)}$. With the same way as in Case 1, by Lemma \ref{lem:GHW11},  we have $\eta_{0}^{(2,r)}=\tfrac{-1-q}{2}$ or $\eta_{0}^{(2,r)}=\tfrac{q-1}{2}$, then
$$T(a,b)=\left\{\begin{array}{llll}
0, & {\rm if} \,\,\, p\equiv3\,\,\,({\rm mod}\,\,\, 4) \,\,\, {\rm and}\,\,\, s\,\,\, {\rm is}\,\,\,  {\rm odd},\\
h, & {\rm if} \,\,\, p\equiv1\,\,\,({\rm mod}\,\,\, 4), \,\,\,{\rm or}\,\,\, p\equiv3\,\,\,({\rm mod}\,\,\, 4) \,\,\, {\rm and}\,\,\, s\,\,\, {\rm is}\,\,\, {\rm even}.
         \end{array}
\right.
$$
Combining Case 1 and Case 2, from Lemma \ref{lem:GHWt},  we have
\begin{equation}\label{eq:Tab82301}
\begin{aligned}
T(a,b)=\left\{\begin{array}{llll}
0,  & \frac{(r-1)(q-1)}{2} & {\rm times,}\\
h, & \frac{(r-1)(q-1)}{2} &{\rm times}
         \end{array}
\right.
\end{aligned}
\end{equation}
 when $(a,b)$ runs over $\big\{ (a,b) \in (\mathbb{F}_r, \mathbb{F}_r)\,:\, \frac{a}{b} \in B \big\}$. From Case 7 in Theorem \ref{thm:2}, we have the same result as (\ref{eq:Tab82301}) when $(a,b)$ runs through $\big\{ (a,b) \in (\mathbb{F}_r, \mathbb{F}_r)\,:\,-\frac{a}{b} \in B\}$.
The desired results then follows.
\end{proof}

\begin{lemma}\label{lem:21}
Let
$$K_{0,1}=\left\{(a,b)\,:\,(a,b) \in (\mathbb{F}_r, \mathbb{F}_r),\,a+b\in C_{0}^{(2,r)}, a-b\in C_{1}^{(2,r)},\pm\frac{a}{b} \notin B\right\}$$
and
$$K_{1,0}=\left\{(a,b)\,:\,(a,b) \in (\mathbb{F}_r, \mathbb{F}_r),\,a-b\in C_{0}^{(2,r)}, a+b\in C_{1}^{(2,r)},\pm\frac{a}{b} \notin B\right\},$$
where $B$ is given in (\ref{eq:B1B2}) and $r=q^2$. Then $T(a,b)=0$ if $(a,b) \in K_{0,1}$ or $(a,b) \in K_{1,0}$, where $T(a,b)$ is given in (\ref{eq:Tab}).
\end{lemma}

\begin{proof}
Let
\begin{equation}\label{vab}
u=a+b\,\,\,\text{and}\,\,\,v=a-b.
\end{equation}
It is clear that $(u,v)$ runs through $(\mathbb{F}_r,\mathbb{F}_r)$ if and only if $(a,b)$ runs through $(\mathbb{F}_r,\mathbb{F}_r)$. This means that for any $(u,v)\in (\mathbb{F}_r,\mathbb{F}_r)$, there exists $(a,b) \in (\mathbb{F}_r,\mathbb{F}_r)$ such that (\ref{vab}) holds. Then
\begin{equation*}
\begin{aligned}
&\left\{\frac{a-b}{a+b}\,:\,(a,b) \in (\mathbb{F}_r, \mathbb{F}_r),\,a+b\in C_{0}^{(2,r)}, \,\,a-b\in C_{1}^{(2,r)}\right\}\\
=&\left\{\frac{a-b}{a+b}\,:\,(a,b) \in (\mathbb{F}_r, \mathbb{F}_r),\,a-b\in C_{0}^{(2,r)},\,\, a+b\in C_{1}^{(2,r)}\right\}\\
=&\left\{\alpha^{2i+1}\,:\,i=0,1,\ldots,\frac{r-3}{2}\right\}\\
=&\left\{C_{2i+1}^{(q+1,r)}\,:\, i=0,1,\ldots,\frac{q-1}{2}\right\}.
\end{aligned}
\end{equation*}
By the definition of $B$, we know that  $\frac{a-b}{a+b}=-g y$ if $-\frac{a}{b} \in B$ and $\frac{a-b}{a+b}=-(g y)^{-1}$ if $\frac{a}{b} \in B$ for $y \in \mathbb{F}_q^*$. Then
$$\left\{\frac{a-b}{a+b}\,:\,(a,b) \in (\mathbb{F}_r, \mathbb{F}_r),\, \pm\frac{a}{b} \in B\right\}=\left\{C_{\frac{q-1}{h}}^{(q+1,r)}\cup C_{q+1-\frac{q-1}{h}}^{(q+1,r)}\right\}.$$
Hence,
we have
\begin{equation}\label{even}
\begin{split}
\left\{\frac{a-b}{a+b}\,:\,(a,b) \in K_{0,1}\right\}&=\left\{\frac{a-b}{a+b}\,:\,(a,b) \in K_{1,0}\right\}\\
&=\left\{C_{2i+1}^{(q+1,r)}\,:\, i=0,\ldots,\frac{q-1}{2}\setminus \left\{\frac{q-1-h}{2h},\frac{qh-q+1}{2h}\right\}\right\}.
\end{split}
\end{equation}

We show the desired result from the following two cases.

\noindent{\bf Case 1:} $p\equiv 3 \pmod 4$ and $s$ is odd.
If $a+b\in C_{0}^{(2,r)}$, $a-b\in C_{1}^{(2,r)}$ and $\pm \frac{a}{b} \notin B$, it is easy to get that $a\neq \pm b\neq 0$. From Case 8 of Theorem \ref{thm:2}, we have
\begin{equation*}
T(a,b)=\tfrac{h}{q^2}(2+4\eta_{0}^{(2,r)})+\tfrac{h}{q^2}(t-l)(1+2\eta_{0}^{(2,r)}),
\end{equation*}
where $t$ and $l$ are given in (\ref{eq:072701}).
By Lemma \ref{lem:GHW11}, we have $\eta_{0}^{(2,r)}=\tfrac{q-1}{2}$. Then
\begin{equation}\label{eqTab}
T(a,b)=\tfrac{h}{q}(2+t-l).
\end{equation}
It is clear that $a+b+(a-b)gy=(a+b)(1+\tfrac{a-b}{a+b}g y)$ and $a-b+(a+b) g y=(a-b)(1+\tfrac{a+b}{a-b}g y)$, then from (\ref{eq:072701}), we know that $t$ and $l$ can be expressed as
\begin{equation}\label{eq:072302}
t=\left|\{1+\tfrac{a-b}{a+b}gy:\,y\in \mathbb{F}_q^*\} \cap C_{0}^{(2,r)}\right|\,\,\, \text{and}\,\,\,
l=\left|\{1+\tfrac{a+b}{a-b}gy:\,y\in \mathbb{F}_q^*\} \cap C_{1}^{(2,r)}\right|
\end{equation}
since $a+b\in C_{0}^{(2,r)}$  and $a-b\in C_{1}^{(2,r)}$. From (\ref{even}) we can assume that $\frac{a-b}{a+b}=\alpha^{2i+1}\langle\alpha^{q+1}\rangle$ and $i=0,\ldots,\frac{q-1}{2}\setminus \{\frac{q-1-h}{2h},\frac{qh-q+1}{2h}\}$, then
\begin{equation}\label{eq:07230201}
\left\{\tfrac{a-b}{a+b}g y\,:\,y\in \mathbb{F}_q^*\right\}=\alpha^{2i+1+\frac{q-1}{h}}\langle\alpha^{q+1}\rangle=C_{2i+1+\frac{q-1}{h}}^{(q+1,r)}
\end{equation}
and
\begin{equation}\label{eq:072302027}
\left\{\tfrac{a+b}{a-b}g y\,:\,y\in \mathbb{F}_q^*\right\}=\alpha^{\frac{q-1}{h}-2i-1}\langle\alpha^{q+1}\rangle=C_{\frac{q-1}{h}-2i-1}^{(q+1,r)}=C_{q-2i+\frac{q-1}{h}}^{(q+1,r)}.
\end{equation}
From (\ref{eq:072302})-(\ref{eq:072302027}) we have
\begin{equation}\label{eq:0820}
t=\left|(1+C_{2i+1+\frac{q-1}{h}}^{(q+1,r)}) \cap C_{0}^{(2,r)}\right|=(2i+1+\frac{q-1}{h},0)^{(q+1,2,r)}
\end{equation}
and
\begin{equation}\label{eq:082001}
l=\left|(1+C_{q-2i+\frac{q-1}{h}}^{(q+1,r)}) \cap C_{1}^{(2,r)}\right|=(q-2i+\frac{q-1}{h},1)^{(q+1,2,r)}.
\end{equation}
For any $(a,b) \in K_{0,1}$, we know that $T(a,b)$ must be an integer. Then from (\ref{eqTab}), there exist integers $k_i$ such that
$$(2i+1+\frac{q-1}{h},0)^{(q+1,2,r)}-(q-2i+\frac{q-1}{h},1)^{(q+1,2,r)}+2=k_iq$$ for $i=0,\ldots,\frac{q-1}{2}\setminus \{\frac{q-1-h}{2h},\frac{qh-q+1}{2h}\}$.

In Lemma \ref{lem:GHWs7201}, we have proved that $(\tfrac{2(q-1)}{h},0)^{(q+1,2,r)}=\tfrac{q-3}{2}$. For the convenience of narration, we assume that $h=q-1$ in the following.
In this case, we have $(2,0)^{(q+1,2,r)}=\tfrac{q-3}{2}$, then
$$(q-1,1)^{(q+1,2,r)}=q-1-(q-1,0)^{(q+1,2,r)}=\tfrac{q+1}{2}.$$ Let $i=1$, we have $(4,0)^{(q+1,2,r)}-(q-1,1)^{(q+1,2,r)}+2=k_1q$. It is obvious that this equality holds if and only if $k_1=0$ and $(4,0)^{(q+1,2,r)}=\frac{q-3}{2}$ since $(q-1,1)^{(q+1,2,r)}=\tfrac{q+1}{2}$.
Let $i=2$, we have $(6,0)^{(q+1,2,r)}-(q-3,1)^{(q+1,2,r)}+2=k_2q$. Then this equality holds if and only if $k_2=0$ and $(6,0)^{(q+1,2,r)}=\frac{q-3}{2}$ since $(q-3,1)^{(q+1,2,r)}=\tfrac{q+1}{2}$.

Continue this work, we have
$(2i,0)^{(q+1,2,r)}=\tfrac{q-3}{2}$ and $(q+1-2i,1)^{(q+1,2,r)}=\tfrac{q+1}{2}$ for $i=1,\ldots, \frac{q-1}{2}$. From (\ref{eqTab}), (\ref{eq:0820}) and  (\ref{eq:082001}), then we have $T(a,b)=0$. When $h$ is equal to the other values, by the same way, we also have $T(a,b)=0$.

If $a-b\in C_{0}^{(2,r)}$,  $a+b\in C_{1}^{(2,r)}$ and $\pm \frac{a}{b} \notin B$,  then $a\neq \pm b\neq 0$. From Case 8 of Theorem \ref{thm:2}, we have
\begin{equation*}
T(a,b)=\tfrac{h}{q^2}(-2-4\eta_{0}^{(2,r)})+\tfrac{h}{q^2}(t-l)(1+2\eta_{0}^{(2,r)})=\tfrac{h}{q}(t-l-2)
\end{equation*}
since $\eta_{0}^{(2,r)}=\tfrac{-1+q}{2}$, where $t$ and $l$ are given in (\ref{eq:072701}). With an analysis similar as above,  we can obtain $T(a,b)=0$.

\noindent{\bf Case 2:} $p\equiv 3 \pmod 4$ and $s$ is even, or $p\equiv 1 \pmod 4$. In this case, from Case 8 of Theorem \ref{thm:2}, we have
\begin{equation*}
T(a,b)=\tfrac{h}{q}(l-t-2)\,\,\,{\rm or}\,\,\,T(a,b)=\tfrac{h}{q}(l-t+2).
\end{equation*}
With the same way as in Case 1, we have
$T(a,b)=0$.

The desired result then follows.
\end{proof}

\begin{remark}\label{rem:19}
In Lemma \ref{lem:21}, we proved that $(2i,0)^{(q+1,2,q^2)}=\tfrac{q-3}{2}$ for $1\leq i\leq \frac{q-1}{2}$, i.e., the set $\{1+\alpha^{2i}\langle \alpha ^{q+1}\rangle\}$
has $\tfrac{q-3}{2}$ square elements in $\mathbb{F}_{q^2}$ for $1\leq i\leq \frac{q-1}{2}$.
\end{remark}

\begin{lemma}\label{lem:22820}
Let
$$K_{0,0}=\left\{(a,b)\,:\,(a,b) \in (\mathbb{F}_r, \mathbb{F}_r),\,a+b\in C_{0}^{(2,r)},\, a-b\in C_{0}^{(2,r)},\,\pm\frac{a}{b} \notin B\right\}$$
and
$$K_{1,1}=\left\{(a,b)\,:\,(a,b) \in (\mathbb{F}_r, \mathbb{F}_r),\,a-b\in C_{1}^{(2,r)},\, a+b\in C_{1}^{(2,r)},\,\pm\frac{a}{b} \notin B\right\},$$
where $B$ is given in (\ref{eq:B1B2}) and $r=q^2$. Then $T(a,b)=0$ if $(a,b) \in K_{0,0}$ or $(a,b) \in K_{1,1}$, where $T(a,b)$ is given in (\ref{eq:Tab}).
\end{lemma}

\begin{proof}
With discussions similar to Lemma \ref{lem:21},  we have
\begin{equation*}
\begin{aligned}
&\left\{\frac{a-b}{a+b}\,:\,(a,b) \in (\mathbb{F}_r, \mathbb{F}_r),\,a+b\in C_{0}^{(2,r)}, a-b\in C_{0}^{(2,r)}\right\}\\
=&\left\{\frac{a-b}{a+b}\,:\,(a,b) \in (\mathbb{F}_r, \mathbb{F}_r),\,a+b\in C_{1}^{(2,r)}, a-b\in C_{1}^{(2,r)}\right\}\\
=&\left\{C_{2i}^{(q+1,r)}\,:\, i=0,1,\ldots,\frac{q-1}{2}\right\}
\end{aligned}
\end{equation*}
and
$$\left\{\frac{a-b}{a+b}\,:\,(a,b) \in (\mathbb{F}_r, \mathbb{F}_r),\, \pm\frac{a}{b} \in B\right\}=\left\{C_{\frac{q-1}{h}}^{(q+1,r)}\cup C_{q+1-\frac{q-1}{h}}^{(q+1,r)}\right\}.$$
Since $\tfrac{q-1}{h} \equiv 1 \pmod 2$, we obtain that $q+1-\frac{q-1}{h}$ is odd. Then
\begin{equation}\label{even1}
\left\{\frac{a-b}{a+b}\,:\,(a,b) \in K_{0,0}\right\}=\left\{\frac{a-b}{a+b}\,:\,(a,b) \in K_{1,1}\right\}=\left\{C_{2i}^{(q+1,r)}\,:\, i=0,1,\ldots,\frac{q-1}{2}\right\}.
\end{equation}
We show the desired result from the following two cases.

\noindent{\bf Case 1:} $p\equiv 3 \pmod 4$ and $s$ is odd.  If $a+b\in C_{0}^{(2,r)}$, $a-b\in C_{0}^{(2,r)}$ and $\pm \frac{a}{b} \notin B$, with an analysis similar to Lemma \ref{lem:21},  we have
\begin{equation}\label{eqt}
T(a,b)=\tfrac{h}{q^2}(t-l)(1+2\eta_{0}^{(2,r)})=\tfrac{h}{q}(t-l),
\end{equation}
where
\begin{equation*}
t=\left|(1+C_{2i+1}^{(q+1,r)}) \cap C_{0}^{(2,r)}\right|=(2i+1,0)^{(q+1,2,r)}
\end{equation*}
and
\begin{equation*}
l=\left|(1+C_{q+2(q-1)/h-2i}^{(q+1,r)}) \cap C_{0}^{(2,r)}\right|=(q+\tfrac{2(q-1)}{h}-2i,0)^{(q+1,2,r)}
\end{equation*}
for $i=0,1,\ldots,\frac{q-1}{2}$.
Since $0\leq t \leq q-1$, $0\leq l \leq q-1$, we obtain $-(q-1)\leq t-l \leq q-1$.
By the definition of $T(a,b)$, we know that $T(a,b)$ is an integer. Hence, $q \, | \, t-l$, i.e.,
\begin{equation}\label{eq:tl80}
t=l=(2i+1,0)^{(q+1,2,r)}=(q+\tfrac{2(q-1)}{h}-2i,0)^{(q+1,2,r)}
\end{equation} for any $i=0,1,\ldots,\frac{q-1}{2}$.
%
%
 Hence, from (\ref{eqt}) we know $T(a,b)=0$.

If $a+b\in C_{1}^{(2,r)}$, $a-b\in C_{1}^{(2,r)}$ and $\pm \frac{a}{b} \notin B$, from Case 8 of Theorem \ref{thm:2}, we have
\begin{equation}\label{eqtt}
T(a,b)=\tfrac{h}{q^2}(t-l)(1+2\eta_{0}^{(2,r)})=\tfrac{h}{q}(t-l).
\end{equation}
 With an analysis similar as above,  we can obtain that $T(a,b)=0$.

\noindent{\bf Case 2:} $p\equiv 3 \pmod 4$ and $s$ is even, or $p\equiv 1 \pmod 4$.
In this case, from Case 8 of Theorem \ref{thm:2}, we have
\begin{equation*}
T(a,b)=-\tfrac{h}{q}(t-l)=\tfrac{h}{q}(l-t).
\end{equation*}
By an analysis similar to Case 1, it is easy to get that $T(a,b)=0$.

The desired result then follows.
\end{proof}


From Lemma \ref{lem:21} and Lemma \ref{lem:22820}, we have the following result.

\begin{lemma}\label{lem:GHWs}
Let $B$ be given in (\ref{eq:B1B2}) and $r=q^2$. If $a\neq\pm b\neq 0$ and $\pm\frac{a}{b} \notin B$, then $T(a,b)=0$.
\end{lemma}

With above preparations, we now give the symbol-pair weight distribution of $\C_{(q,2,h,2)}$.

\begin{theorem}\label{thm:4}
Let $\C_{(q,m,h,e)}$ be the cyclic code defined in (\ref{cqmhe}).
Let $\tfrac{q-1}{h} \equiv 1 \pmod 2$, then $\C_{(q,2,h,2)}$ is an $[h(q+1),4]$ code with the symbol-pair weight enumerator
\begin{equation}
\begin{aligned}
1+(r-1)z^{h(q-1)}+(r-1)(q-1)z^{hq}+(r-1)(r+1-q)z^{h(q+1)}.\nonumber
\end{aligned}
\end{equation}
\end{theorem}

\begin{proof}
Combining Theorem \ref{thm:2}, Lemma \ref{lem:GHWs720}, Lemma \ref{lem:GHWs7201} and Lemma \ref{lem:GHWs}, it is easy to get that
\begin{equation}\label{eq:0715015}
\begin{aligned}
T(a,b)=\left\{\begin{array}{llll}
h(q+1),    &  1 & {\rm time,}\\
2h ,    &(r-1)& {\rm times,}\\
h ,     &(r-1)(q-1)& {\rm times,}\\
0 ,     &(r-1)(r+1-q)&  {\rm times.}
         \end{array}
\right.
\end{aligned}
\end{equation}
For any $(a,b) \neq (0,0)$, we have that $T(a,b)$ is less than the length of the code. Then the dimension of the code $\C_{(q,2,h,2)}$ in (\ref{cqmhe}) is $4$.
From  (\ref{eq:cab}) and (\ref{eq:0715015}), the symbol-pair weight enumerator of $\C_{(q,2,h,2)}$ can be calculated. The desired result then follows.
\end{proof}

\begin{example}
Let $q=3$, $m=2$, $e=2$ and $h=2$. Then $\C_{(q,m,h,e)}$ is a $[8,4]$ code over $\bF_{3}$ with the symbol-pair weight enumerator $1+8z^4+16z^6+56z^8$. The result is verified by Magma programs.
\end{example}

\begin{example}
Let $q=13$, $m=2$, $e=2$ and $h=4$. Then $\C_{(q,m,h,e)}$ is a $[56,4]$ code over $\bF_{13}$ with the symbol-pair weight enumerator $1+168z^{48}+2016z^{52}+26376z^{56}$. The result is verified by Magma programs.
\end{example}

\subsection{The case  $\tfrac{q-1}{h} \equiv 0 \pmod 2$}

In this subsection, we always assume that $\tfrac{q-1}{h} \equiv 0 \pmod 2$. The possible symbol-pair weights of
$\mathcal{C}_{(q,m,h,2)}$ are given as follows.

\begin{theorem}\label{thm:3}
Let $gcd(m,\tfrac{2(q-1)}{h})=2$. When $(a,b)$ runs through $(\mathbb{F}_r,\mathbb{F}_r)$, then the set of all possible symbol-pair weights of $\C_{(q,m,h,2)}$ is
\begin{equation}\label{eq:0822}
\begin{split}
&\Big\{ 0, \, hq^{m-2}(q+1) \pm hq^{\tfrac{m}{2}-2}(3-q+2(\tfrac{q-1}{h},0)^{(\tfrac{r-1}{q-1},2,r)}),\, hq^{m-1} \pm hq^{\tfrac{m}{2}-1},\\
&hq^{m-2}(q+\tfrac{1}{2}) \pm hq^{\tfrac{m}{2}-2}(\tfrac{3}{2}+(\tfrac{2(q-1)}{h},0)^{(\tfrac{r-1}{q-1},2,r)}),\,
hq^{m-2}(q+1)+hq^{\tfrac{m}{2}-2}(t+l+\xi-q)
\Big\},
\end{split}
\end{equation}
where $0\leq t,l\leq q-1$ and $\xi=3, \pm 1$. Moreover, the code achieves the largest minimum symbol-pair distance compared with its minimum Hamming distance, i.e., $$d_p(\C_{(q,m,h,2)})=2d_H(\C_{(q,m,h,2)})=hq^{m-1} - hq^{\tfrac{m}{2}-1}.$$
\end{theorem}

\begin{proof}
In order to obtain the desired results, from (\ref{eq:cab}) we only need to consider the possible value of $T(a,b)$ for $(a,b)$ running through $(\mathbb{F}_r,\mathbb{F}_r)$.

By the definition of cyclotomic class, it is clear that $C_{(q-1)/h}^{(2,r)}=C_{0}^{(2,r)}$ if $\tfrac{q-1}{h} \equiv 0\pmod 2$. Since $e=2$, by the definition of $\beta$, we have $\beta=-1$. Then from (\ref{eq:Tab}) we have
\begin{equation}\label{Tab0711}
\begin{aligned}
T(a,b)&=\tfrac{h(r-1)}{q^2(q-1)}+\tfrac{2h}{q^2}\Big[\sum_{x \in C_{0}^{(2,r)}}\chi(x(a+b))+\sum_{x \in C_{0}^{(2,r)}}\chi(x(a-b))\Big]\\
&+\tfrac{h}{q^2}\Big[\sum_{x \in C_{0}^{(2,r)}}\sum_{y \in \bF_{q}^\ast}\chi(x(a+b+(a- b)gy))+\sum_{x \in C_{0}^{(2,r)}}\sum_{y \in \bF_{q}^\ast}\chi(x(a-b+(a+b
)gy))\Big]\\
\end{aligned}
\end{equation}
since $d=\gcd(m,\tfrac{2(q-1)}{h})=2$.
We now prove this theorem from the following six cases.

\noindent{\bf Case 1:} $a=0$ and $b=0$. From the definition of $T(a,b)$,  we have $T(a,b)=\tfrac{h(r-1)}{q-1}.$

\noindent{\bf Case 2:} $a\neq0$ and $b=0$. In this case, (\ref{Tab0711}) can be written as
\begin{equation*}
\begin{aligned}
T(a,b)&=\tfrac{h(r-1)}{q^2(q-1)}+\tfrac{4h}{q^2}\sum_{x \in C_{0}^{(2,r)}}\chi(ax)+\tfrac{2h}{q^2}\sum_{x \in C_{0}^{(2,r)}}\sum_{y \in \bF_{q}^\ast}\chi(xa(1+gy)).\nonumber
\end{aligned}
\end{equation*}
It is clear that $1+gy\neq0$ for $y\in\bF_{q}^\ast$ since $g\notin\bF_{q}$. Then
\begin{equation*}
\{gy:y\in\mathbb{F}_q^*\}=\alpha^{\frac{q-1}{h}}\langle\alpha^{\frac{r-1}{q-1}}\rangle=C_{\frac{q-1}{h}}^{(\frac{r-1}{q-1},r)}.
\end{equation*}
By the definition of generalized cyclotomic numbers, we have
\begin{equation*}
\begin{aligned}
&|\left\{1+gy\,:\,y \in \bF_{q}^{\ast}  \right\} \cap C_{i}^{(2,r)}|=(\tfrac{q-1}{h},i)^{(\frac{r-1}{q-1},2,r)},
\end{aligned}
\end{equation*}
where $i=0$ or $i=1$.
If $a \in C_0^{(2,r)}$, we have
\begin{equation*}
\begin{aligned}
T(a,b)&=\tfrac{h(r-1)}{q^2(q-1)}+\tfrac{4h}{q^2}\eta_{0}^{(2,r)}+
\tfrac{2h}{q^2}\Big[(\tfrac{q-1}{h},0)^{(\tfrac{r-1}{q-1},2,r)}\eta_{0}^{(2,r)}+(\tfrac{q-1}{h},1)^{(\tfrac{r-1}{q-1},2,r)}\eta_{1}^{(2,r)}\Big]\\
&=\tfrac{h(r-1)}{q^2(q-1)}+\tfrac{4h}{q^2}\eta_{0}^{(2,r)}+\tfrac{2h}{q^2}\Big[(\tfrac{q-1}{h},0)^{(\tfrac{r-1}{q-1},2,r)}\eta_{0}^{(2,r)}+
(q-1-(\tfrac{q-1}{h},0)^{(\tfrac{r-1}{q-1},2,r)})(-1-\eta_{0}^{(2,r)})\Big]\\
&=\tfrac{h(q^{m-2}-1)}{q-1}+\tfrac{h}{q^2}(1+2\eta_{0}^{(2,r)})(3-q+2(\tfrac{q-1}{h},0)^{(\tfrac{r-1}{q-1},2,r)}).
\end{aligned}
\end{equation*}
In the second equality, we used the fact that  $\eta_{0}^{(2,r)}+\eta_{1}^{(2,r)}=-1$, which is given in Lemma \ref{lem:GHW11}.
If $a \in C_1^{(2,r)}$, by the same way as above, we have $$T(a,b)=\tfrac{h(q^{m-2}-1)}{q-1}+\tfrac{h}{q^2}(1+2\eta_{1}^{(2,r)})(3-q+2(\tfrac{q-1}{h},0)^{(\tfrac{r-1}{q-1},2,r)}).$$

\noindent{\bf Case 3:} $a=0$ and $b\neq0$. With
discussions similar to Case 2, we have
$$T(a,b)=\tfrac{h(q^{m-2}-1)}{q-1}+\tfrac{h}{q^2}(1+2\eta_{i}^{(2,r)})(3-q+2(\tfrac{q-1}{h},0)^{(\tfrac{r-1}{q-1},2,r)})$$
if $b \in C_i^{(2,r)}$, where $i=0$ or $i=1$.

\noindent{\bf Case 4:} $a= b\neq0$ or  $a= -b\neq0$.
With an analysis similar to Case 4 of Theorem \ref{thm:2}, we have
$$T(a,b)=\tfrac{h}{q}(\tfrac{r-1}{q-1}+2\eta_i^{(2,r)})$$
if $a \in C_i^{(2,r)}$, where $i=0$ or $i=1$.

\noindent{\bf Case 5:} $\frac{a}{b} \in U$ or $-\frac{a}{b} \in U$, where
\begin{equation}\label{eq:B1B2b}
\begin{split}
U=\left\{\tfrac{gy-1}{1+gy}:y \in \bF_{q}^\ast\right\}.
\end{split}
\end{equation}
With discussions similar to Case 6 of Theorem \ref{thm:2}, we have
$$T(a,b)=\tfrac{h(q^{m-2}-1)}{q-1}+\tfrac{h}{q^2}(\tfrac{q^m}{2}+\tfrac{3}{2}+3\eta_{i}^{(2,r)}+(\tfrac{2(q-1)}{h},0)^{(\tfrac{r-1}{q-1},2,r)}(1+2\eta_{i}^{(2,r)}))$$
if $a-b\in C_{i}^{(2,r)}$, where $i=0$ or $i=1$.

\noindent{\bf Case 6:} $a\neq\pm b\neq 0$ and $\pm\frac{a}{b} \notin U$.
With an analysis similar to the Case 8 of Theorem \ref{thm:2}, we have
\begin{equation*}
\begin{aligned}
T(a,b)=\begin{cases}
\tfrac{h(q^{m-2}-1)}{q-1}+\tfrac{h}{q^2}(1+2\eta_{0}^{(2,r)})(t+l+3-q),&   {\rm if} \,\,\, a+b\in C_0^{(2,r)}\,\,\,{\rm and}\,\,\,a-b\in C_0^{(2,r)},\\
\tfrac{h(q^{m-2}-1)}{q-1}+\tfrac{h}{q^2}(1+2\eta_{0}^{(2,r)})(t+l-q-1), & {\rm if}  \,\,\, a+b\in C_1^{(2,r)}\,\,\,{\rm and}\,\,\,a-b\in C_1^{(2,r)},\\
\tfrac{h(q^{m-2}-1)}{q-1}+\tfrac{h}{q^2}(1+2\eta_{0}^{(2,r)})(t+l+1-q), &  {\rm if}  \,\,\, a+b\in C_i^{(2,r)}\,\,\,{\rm and} \,\,\,a-b\in C_{i+1}^{(2,r)},
\end{cases}
\end{aligned}
\end{equation*}
where
$$t=|\{a+b+(a-b)gy\,:\,y\in \mathbb{F}_q^*\} \cap C_{0}^{(2,r)}|\,\, \text{and}\,\,l=|\{a-b+(a+b)gy\,:\,y\in \mathbb{F}_q^*\} \cap C_{0}^{(2,r)}|.$$
Obviously, we have $0\leq t,l\leq q-1$.

Combining all the cases, let $M_i=\tfrac{q^m}{2}+\tfrac{3}{2}+3\eta_{i}^{(2,r)}+(\tfrac{2(q-1)}{h},0)^{(\tfrac{r-1}{q-1},2,r)}(1+2\eta_{i}^{(2,r)})$ for $i=0$ or $i=1$, $M_2=2(\tfrac{q-1}{h},0)^{(\tfrac{r-1}{q-1},2,r)}$ and $U$ be defined in (\ref{eq:B1B2b}), we have
\begin{equation}\label{eq:0825014}
T(a,b)=\tfrac{h(q^{m-2}-1)}{q-1}+\tfrac{h}{q^2}T'(a,b),
\end{equation}
where
\begin{equation}\label{eq:0715014}
\begin{aligned}
T'(a,b)=\begin{cases}
q^{m}(q+1),  & {\rm if }\,\,\,a=0\,\,\, {\rm and}\,\,\,b=0,\\
(1+2\eta_{i}^{(2,r)})(3-q+M_2), &{\rm if }\,\,\,a \in C_i^{(2,r)}\,\,\, {\rm and}\,\,\, b=0,\,\,\, {\rm or} \,\,\, b \in C_i^{(2,r)}\,\,\, and\,\,\, a=0,\\
q^m+q+2q\eta_{i}^{(2,r)},  &{\rm if} \,\,\, a= b \in C_i^{(2,r)}, \,\,\, {\rm or} \,\,\,a= -b \in C_i^{(2,r)},\\
M_i,&{\rm if} \,\,\,  \tfrac{a}{b}\in U \,\,\, {\rm and}\,\,\, a-b\in C_i^{(2,r)},\,\, {\rm or} \,
-\tfrac{a}{b}\in U \,\,\, {\rm and}\,\,\, a-b\in C_i^{(2,r)},\\
(1+2\eta_{0}^{(2,r)})(t+l+3-q),  & {\rm if} \,\,\,\, a\neq0, \,\,\pm\tfrac{a}{b}\notin U, \,\,\,  a+b\in C_0^{(2,r)}\,\,\,{\rm and} \,\,\,a-b\in C_0^{(2,r)},\\
(1+2\eta_{0}^{(2,r)})(t+l-q-1),  & {\rm if} \,\,\,\, a\neq0,\,\,\,\pm\tfrac{a}{b}\notin U, \,\,\, a+b\in C_1^{(2,r)}\,\,\,{\rm and} \,\,\,a-b\in C_1^{(2,r)},\\
(1+2\eta_{0}^{(2,r)})(t+l+1-q),  & {\rm if} \,\,\,\, a\neq0,\,\,\,\pm\tfrac{a}{b}\notin U,  \,\,\, a+b\in C_i^{(2,r)}\,\,\,{\rm and} \,\,\,a-b\in C_{i+1}^{(2,r)}\\
\end{cases}
\end{aligned}
\end{equation}
for $i=0$ or $i=1$, and the indices are taken modulo $2$.
 Then from Lemma \ref{lem:GHW11} and (\ref{eq:cab}), the possible symbol-pair weights of $\C_{(q,m,h,2)}$ can be given and we show this result in (\ref{eq:0822}).

From (\ref{eq:0822}), it is obvious that the minimum possible symbol-pair weight of $\C_{(q,m,h,2)}$ is $hq^{m-1} - hq^{\tfrac{m}{2}-1}$. In \cite[Theorem 6]{CMa2010}, the authors showed that $d_H(\C_{(q,m,h,2)})=\frac{1}{2}(hq^{m-1} - hq^{\tfrac{m}{2}-1})$, then $d_p(\C_{(q,m,h,2)})=hq^{m-1} - hq^{\tfrac{m}{2}-1}$ since $ d_p(\C) \leq 2d_H(\C)$ for any code $\C$. The desired result then follows.
\end{proof}

Similar to the proof of the case $\tfrac{q-1}{h} \equiv 1 \pmod 2$, for the special case $m=2$, we can show that $\mathcal{C}_{(q,2,h,2)}$ is a three symbol-pair weight code and determine its symbol-pair weight distribution. Details of the proof are omitted here.

\begin{theorem}\label{thm:5}
Let $\C_{(q,m,h,e)}$ be the cyclic code defined in (\ref{cqmhe}).
Let $e=m=2$ and $\tfrac{q-1}{h} \equiv 0 \pmod 2$, then $\C_{(q,2,h,2)}$ is an $[h(q+1),4]$ code with the symbol-pair weight enumerator
\[
1+(r-1)z^{h(q-1)}+(r-1)(q-1)z^{hq}+(r-1)(r+1-q)z^{h(q+1)}.
\]
\end{theorem}

\begin{example}
Let $q=9$, $m=2$, $e=2$ and $h=4$. Then $\C_{(q,m,h,e)}$ is a $[40,4]$ code over $\bF_{9}$ with the symbol-pair weight enumerator $1+80z^{32}+640z^{36}+5840z^{40}$. The result is verified by Magma programs.
\end{example}

\begin{example}
Let $q=17$, $m=2$, $e=2$ and $h=4$. Then $\C_{(q,m,h,e)}$ is a $[72,4]$ code over $\bF_{17}$ with the symbol-pair weight enumerator $1+288z^{64}+4608 z^{68}+78624z^{72}$. The result is verified by Magma programs.
\end{example}

\section{The pair weight distribution of the three class of cyclic codes}\label{sec3}

In this section, we investigate the punctured code of $\C_{(q,m,h,e)}$, which is derived from this code by deleting some coordinates of
codewords. Some new linear codes are obtained.

Let $e=2$, from (\ref{cqmhe}) the codeword of $\C_{(q,m,h,e)}$ can be expressed as
$${\bf c}(a,b)=({\rm Tr}_q^r(ag^0+b(- g)^0),{\rm Tr}_q^r(ag+b(- g)),\ldots,{\rm Tr}_q^r(ag^{n-1}+b(- g)^{n-1}))$$
since $\beta=\alpha^{\frac{r-1}{2}}=-1$. If $n\equiv 0 \pmod 4$, then $ (-g)^{\frac{n}{2}}=\alpha^{\frac{r-1}{2}}=-1$, which implies that
$${\rm Tr}_q^r(ag^i+b(- g)^i)=-{\rm Tr}_q^r(ag^{i+\frac{n}{2}}+b(- g)^{i+\frac{n}{2}}),$$
where $0\leq i\leq \frac{n-2}{2}$. Let $${\bf c'}(a,b)=({\rm Tr}_q^r(ag^0+b(- g)^0),{\rm Tr}_q^r(ag+b(- g)),\ldots,{\rm Tr}_q^r(ag^{\frac{n-2}{2}}+b(- g)^{\frac{n-2}{2}})).$$
Then $ \C_{(q,m,h,e)}$ can be expressed as
$ \C_{(q,m,h,e)}=\left\{({\bf c}'(a,b),-{\bf c}'(a,b))\,:\,a,b \in \bF_{r}\right\}.$
Let \begin{equation}\label{eq:0821}
\C'_{(q,m,h,e)}=\left\{{\bf c}'(a,b)\,:\,a,b \in \bF_{r}\right\}.
\end{equation}
It is clear that the symbol-pair weights of the codewords in $\C'_{(q,m,h,e)}$ are the half of $\C_{(q,m,h,e)}$, and the dimension of $\C'_{(q,m,h,e)}$ is equal to  the dimension of $\C_{(q,m,h,e)}$. From Section III and Section IV, we have the following results, directly.

\begin{theorem}\label{thm:6}
Let $q$ be an odd prime power, $\C'_{(q,m,h,e)}$ be the cyclic code defined in (\ref{eq:0821}) and $\tfrac{h(r-1)}{q-1}\equiv 0 \pmod 4$. The following statements hold.
\begin{enumerate}
\item[(1)]
If $m>2$ and $\gcd\left(m,\tfrac{2(q-1)}{h}\right)=1$, then $\C'_{(q,m,h,e)}$ is an $\left[\tfrac{h(r-1)}{2(q-1)},2m\right]$ code with the symbol-pair weight enumerator
\[
 1+e(r-1)z^{\frac{hq^{m-2}(eq+e-2)}{2e}}+e(r-1)(q-1)z^{\frac{hq^{m-2}(eq+e-1)}{2e}}+(r-1)(r+1-eq)z^{\frac{hq^{m-2}(q+1)}{2}}.
\]
\item[(2)] If $m\geq2$ and $\gcd\left(m,\tfrac{2(q-1)}{h}\right)=2$, then $\C'_{(q,m,h,2)}$ is an $\left[\tfrac{h(r-1)}{2(q-1)},2m\right]$ code with the minimum symbol-pair distance $d_p(\C'_{(q,m,h,2)})=\tfrac{(hq^{m-1}-hq^{\frac{m}{2}-1})}{2}$. Moreover, $\C'_{(q,2,h,2)}$ is an $\left[\tfrac{h(q+1)}{2},4\right]$ code with the symbol-pair weight enumerator
\[
1+(r-1)z^{\frac{h(q-1)}{2}}+(r-1)(q-1)z^{\frac{hq}{2}}+(r-1)(r+1-q)z^{\frac{h(q+1)}{2}}.
\]
In addition, $\C'_{(q,2,2,2)}$ is an MDS symbol-pair code.
\end{enumerate}
\end{theorem}

\begin{example}
Let $q=9$, $m=2$, $e=2$ and $h=4$. Then $\C'_{(q,m,h,e)}$ is a $[20,4]$ code over $\bF_{9}$ with the symbol-pair weight enumerator $1+80z^{16}+640z^{18}+5840z^{20}$. The result is verified by Magma programs.
\end{example}

\begin{example}
Let $q=17$, $m=2$, $e=2$ and $h=2$. Then $\C'_{(q,m,h,e)}$ is a $[18,4]$ symbol-pair MDS code over $\bF_{17}$ with the symbol-pair weight enumerator $1+288z^{16}+ 4608 z^{17}+78624z^{18}$. The result is verified by Magma programs.
\end{example}

\section{Conclusion}\label{sec-finals}
In this paper, we considered the cyclic codes defined in \cite{CMa2010} under the symbol-pair metric. The possible symbol-pair weights of these codes are obtained. The considered codes has three symbol-pair weights in some cases. In addition, when $e=2$, the codes achieves the largest minimum symbol-pair distances compared with their minimum Hamming distances.
The main contributions of this paper are the following:
\begin{itemize}
\item The symbol-pair weight distribution of $\C_{(q,m,h,e)}$ was given for the case that $m>2$ is a positive integer and ${\rm gcd}(m,\tfrac{e(q-1)}{h})=1$ (see Theorem \ref{thm:1}).
\item The possible symbol-pair weights and minimum symbol-pair distance of $\C_{(q,m,h,2)}$ were derived for the case that ${\rm gcd}(m,\tfrac{2(q-1)}{h})=2$ (see Theorem \ref{thm:2} and Theorem \ref{thm:3}).
\item  The symbol-pair weight distribution of $\C_{(q,2,h,2)}$ was given (see Theorem \ref{thm:4} and Theorem \ref{thm:5}).
\item The punctured code of $\C_{(q,m,h,e)}$ was considered and a class of MDS symbol-pair codes was obtained (see Theorem \ref{thm:6}).
\item The exact values of  a class of  generalized cyclotomic numbers were determined (see Remark \ref{rem:19}).
\end{itemize}

\end{document}